%% file: main.tex
\documentclass[sigconf,balance=false,nonacm]{acmart}

\input{settings.tex}

\begin{document}

\title{PLAN: Variance-Aware Private Mean Estimation}

\author{Martin Aumüller}
\affiliation{%
	\institution{IT University of Copenhagen}
	\city{Copenhagen}
	\state{}
	\country{Denmark}}
\email{maau@itu.dk}

\author{Christian Janos Lebeda}
\affiliation{%
	\institution{Basic Algorithms Research Copenhagen}
	\institution{IT University of Copenhagen}
	\state{}
	\city{Copenhagen}
	\country{Denmark}}
\email{christian.j.lebeda@gmail.com}

\author{Boel Nelson}
\affiliation{%
	\institution{Aarhus University}
	\state{}
	\city{Aarhus}
	\country{Denmark}}
	\email{boel@cs.au.dk}

\author{Rasmus Pagh}
	\affiliation{%
	\institution{Basic Algorithms Research Copenhagen}
	\institution{University of Copenhagen}
	\state{}
	\city{Copenhagen}
	\country{Denmark}}
	\email{pagh@di.ku.dk}

\input{tex/abstract}


\maketitle

\input{tex/introduction}

\input{tex/background}
\input{tex/algorithm}
\input{tex/analysis}
\input{tex/generic}

\input{tex/experiments}
\input{tex/related-work.tex}

\input{tex/conclusion}

\begin{acks}
Lebeda, Nelson, and Pagh carried out this work at Basic Algorithms Research Copenhagen (BARC), supported by the VILLUM Foundation grant 16582. 
Nelson and Pagh were also supported by Providentia, a Data Science Distinguished Investigator grant from Novo Nordisk Fonden.
\end{acks}

\bibliography{ref}
\bibliographystyle{ACM-Reference-Format}

\input{tex/supplemental}

\end{document}

%% file: settings.tex
\usepackage{algorithm}
\usepackage{algorithmicx}

\usepackage{microtype}
\usepackage{graphicx}
\usepackage{subcaption}
\usepackage{booktabs}
\usepackage{tabularx}

\usepackage{hyperref}
\usepackage{xspace}
\usepackage{cleveref}
\usepackage{enumitem}

\usepackage{multirow}
\usepackage{varwidth}
\theoremstyle{plain}
\newtheorem{theorem}{Theorem}[section]

\newtheorem{lemma}[theorem]{Lemma}
\newtheorem{corollary}[theorem]{Corollary}
\theoremstyle{definition}
\newtheorem{definition}[theorem]{Definition}

\theoremstyle{remark}

\usepackage[disable,textsize=tiny]{todonotes}
\newcommand{\matodo}[1]{\todo[inline, color=green!30]{{MA: #1}}}

\newcommand{\chtodo}[1]{\todo[inline, color=purple!30]{{CL: #1}}}

\newcommand{\algorithmnamelong}{Private Limit Adapted Noise (\algorithmname)\xspace}
\newcommand{\algorithmname}{\textsc{plan}\xspace}
\newcommand{\instanceoptimal}{\textsc{Instance-optimal mean estimation}\xspace}
\newcommand{\instanceoptimallong}{\textsc{Instance-optimal mean estimation} (\instanceoptimalshort)\xspace}
\newcommand{\instanceoptimalshort}{\textsc{iome}\xspace}
\newcommand{\coinpress}{\textsc{CoinPress}\xspace}
\newcommand{\zcdplong}{zero-Concentrated Differential Privacy (zCDP)\xspace}
\newcommand{\zcdp}{zCDP\xspace}

\newcommand{\universevariable}{\ensuremath{M}}
\newcommand{\casethem}{Gaussian A\xspace}
\newcommand{\caseplausible}{Gaussian B\xspace}
\newcommand{\caseshine}{Gaussian C\xspace}
\newcommand{\casebinary}{Binary\xspace}
\newcommand{\casekosarak}{Kosarak\xspace}
\newcommand{\casepos}{POS\xspace}
\newcommand\compactdots{\hbox to 1em{.\hss.\hss.}}
\newcommand{\scalingvec}{\vectorize{s}} 
\newcommand{\privatequantile}{\textsc{PrivQuantile}\xspace}

\newcommand{\elltwowellconcentrated}{$\vectorize{\sigma}$-well concentrated\xspace}

\newcommand{\probability}{\ensuremath{q}}

\newcommand{\ourparagraph}[2]{\noindent\emph{\textbf{#1}#2}.}

\usepackage{listings}
\crefname{lstlisting}{Listing}{Listings}
\usepackage{xcolor}
\definecolor{orangeish}{HTML}{f1a340} 
\definecolor{grayish}{HTML}{f7f7f7} 
\definecolor{purpleish}{HTML}{998ec3} 
\definecolor{blueish}{HTML}{004488}

\newcommand{\basicCodeStyle}{\ttfamily\small}
\newcommand{\keywordCodeStyle}{\bfseries}
\newcommand{\builtInCodeStyle}{\itshape\color{orangeish}}
\newcommand{\typesCodeStyle}{\ttfamily\small\color{blueish}}

\newcommand{\vectorize}[1]{\ensuremath{\boldsymbol{#1}}}

\newcommand{\pessimisticuniverse}{\ensuremath{100\sqrt{d}\max\left\{\vectorize{\sigma}\right\}}}

\lstdefinelanguage{plan}{
	classoffset=0,
	morekeywords={True, False, return, null, if, in, while, do, else, case, break, continue, for, def, log, sqrt, norm},
		keywordstyle=\keywordCodeStyle,
	classoffset=1,
	morekeywords={privQuant, recenter,  scale, clip, noise},
	keywordstyle=\typesCodeStyle,
	classoffset=2,
	morekeywords={estimateStd, rankError, divideBudget, size, scope, center},
	keywordstyle=\builtInCodeStyle,
	classoffset=0,
	morecomment=[s]{/*}{*/},
	morecomment=[l]//,
	morestring=[b]",
	morestring=[b]'
}

%% file: tex/abstract.tex
\begin{abstract}
Differentially private mean estimation is an important building block in privacy-preserving algorithms for data analysis and machine learning.
Though the trade-off between privacy and utility is well understood in the worst case, many datasets exhibit structure that could potentially be exploited to yield better algorithms.
In this paper we present \textit{\algorithmnamelong}, a family of differentially private algorithms for mean estimation in the setting where inputs are independently sampled from a distribution $\mathcal{D}$ over $\mathbf{R}^d$, with coordinate-wise standard deviations $\vectorize{\sigma} \in \mathbf{R}^d$.
Similar to mean estimation under Mahalanobis distance, \algorithmname tailors the shape of the noise to the shape of the data, but unlike previous algorithms the privacy budget is spent non-uniformly over the coordinates.
Under a concentration assumption on $\mathcal{D}$, we show how to exploit skew in the vector $\vectorize{\sigma}$, obtaining a (zero-concentrated) differentially private mean estimate with $\ell_2$ error proportional to $\|\vectorize{\sigma}\|_1$.
Previous work has either not taken $\vectorize{\sigma}$ into account, or measured error in Mahalanobis distance --- in both cases resulting in $\ell_2$ error proportional to $\sqrt{d}\|\vectorize{\sigma}\|_2$, which can be up to a factor $\sqrt{d}$ larger.
To verify the effectiveness of \algorithmname, we empirically evaluate accuracy on both synthetic and real-world data.
\end{abstract}

%% file: tex/introduction.tex
\section{Introduction}\label{sec:introduction}
Differentially private mean estimation is an important building block in many algorithms, notably in for example implementations of private stochastic gradient descent~\cite{abadi_deep_learning_2016, pichapati_adaclip_2019,denisov_improved_2022}.
While differential privacy is an effective and popular definition for privacy-preserving data processing, privacy comes at a cost of accuracy, and striking a good trade-off between utility and privacy can be challenging.
Achieving a good trade-off is especially hard for high-dimensional data, as the required \emph{noise} that ensures privacy increases with the number of dimensions.
Making matters worse, differential privacy operates on a worst-case basis.
Adding noise na\"ively may result in a \emph{one-size-fits-none} noise scale that, although private, fails to give meaningful utility for many datasets.

{\bf Motivating example.}
Consider the following toy example:
take a constant $\probability\in (0,1)$ and let $\lambda \ll 1/d$ be a value close to zero.
We consider a set of vectors $\vectorize{x}^{(1)},\dots,\vectorize{x}^{(n)} \in\mathbf{R}^d$ where, independently,
\begin{align*}
	\vectorize{x}^{(j)} =
\begin{cases}
	(1,\lambda,\dots,\lambda) & \text{with probability } \probability\\
	(0,0,\dots,0) & \text{with probability } 1-\probability \enspace .
\end{cases}
\end{align*}
Though the nominal dimension is $d$, the distribution is ``essentially 1-dimensional''.
Thus, we should be able to release a private estimate of the mean $(\probability, \lambda \probability, \dots ,\lambda \probability)$ with about the same precision as a scalar value with sensitivity~1.
However, previous private mean estimation techniques either:
\begin{itemize}
\item add the same amount of noise to all $d$ dimensions, making the $\ell_2$ norm of the error a factor $\Theta(\sqrt{d})$ larger, or
\item split the privacy budget evenly across the $d$ dimensions, making the noise added to the first coordinate a factor $\Theta(\sqrt{d})$ larger than in the scalar case.
\end{itemize}
Instead, we would like to adapt to the situation, and spend most of the privacy budget on the first coordinate while still keeping the noise on the remaining $d-1$ coordinates low.
Dimension reduction using private PCA~\citep{amin_differentially_2019,dwork_analyze_2014, hardt_noisy_2014} could be used, but we aim for something simpler and more general: to spend more budget on coordinates with higher variance.

We explore the design space of variance-aware budget allocations and show that, perhaps surprisingly, the intuitive approach of splitting the budget across coordinates proportional to their spread  is \emph{not optimal} for $\ell_2$ error, and $\ell_p$ error in more generality. This is in stark contrast to the case of measuring error in Mahalanobis distance for which this allocation is optimal~\citep{brown_covariance-aware_2021}.
Another surprise is that we achieve better bounds than the \emph{instance-optimal} algorithm proposed by Huang et al. in~\cite{huang_instance-optimal_2021} in the case where variances \emph{are skewed}.
Using the idea of variance-aware budgeting, we design a family of algorithms, \algorithmnamelong, that distributes the privacy budget optimally.
For the analysis, we assume that an estimate on the variances is known. 
Since this is often unrealistic in a practical setting, 
we propose and experimentally evaluate different methods to obtain a differentially private estimate on the variances. 
Even using such rough estimates, \algorithmname is competitive to, or better than the state-of-the-art solutions in practice.
The contributions of this work cover both theoretical and practical ground:
\begin{itemize}
	\item We present \algorithmnamelong (\Cref{sec:algorithm}), a family of algorithms for differentially private mean estimation for $\ell_2$ error, and its accompanying privacy and utility analysis (\Cref{sec:analysis}).
	\item We formalize, introduce, and exemplify the notion of \emph{\elltwowellconcentrated} distributions (\Cref{sec:analysis} \& \Cref{sec:well-concentrated}) that allow for the study of variance-aware algorithms for mean estimation.
	\item We generalize the analysis of \algorithmname to hold for arbitrary $\ell_p$ error (\Cref{sec:analysis-generalization}).
	\item We prove that \algorithmname matches the utility of the current state-of-the-art algorithm~\cite{huang_instance-optimal_2021} for differentially private mean estimation up to constant factors in expectation, even if no estimate on the variances is known or if there are correlations in the data. 
	This makes \algorithmname a suitable replacement for mean estimation in general settings without requiring computationally expensive random rotations (\Cref{sec:generic-bounds}).
	\item We implement two instantiations of \algorithmname, for $\ell_1$ and $\ell_2$ error respectively, and empirically evaluate the algorithms on synthetic and real-world datasets (\Cref{sec:experiments}), proposing and evaluating several techniques for differentially private variance estimation (\Cref{app:std-evaluation}).
\end{itemize}

\paragraph{Scope.}
We consider the setting where we have independent samples $\vectorize{x}^{(1)},\dots,\vectorize{x}^{(n)}$ from a distribution $\mathcal{D}$ over $\mathbf{R}^d$, and use $\vectorize{\sigma}\in\mathbf{R}^d$ to denote the vector where $\vectorize{\sigma}_i$ is the standard deviation of the $i$th coordinate of a sample from~$\mathcal{D}$.
Our objective is to estimate the mean $\vectorize{\mu}$ of $\mathcal{D}$ with small $\ell_p$ error.
This is a natural error measure in many settings. For example, $p=1$ is analogous to the mean absolute error, while $p=2$ represents the mean squared error.
Other interesting applications arise if input vectors:
\begin{itemize}
\item represent probability distributions over $\{1,\dots,d\}$, so that $\ell_1$ error corresponds to variation distance, or
\item are binary and model $d$ counting queries, so that small $\ell_1$ and $\ell_2$ error corresponds to small mean error and root mean square error, respectively.
\end{itemize}
Our exposition will focus primarily on the case of $\ell_2$ error, which is most easily compared to previous work. We remark that in the special case of  $\mathcal{D}$ being a Gaussian distribution, \emph{Mahalanobis distance} (\Cref{sec:distance-measures}) is a natural, well-understood~\citep{brown_covariance-aware_2021} error measure.

\algorithmname's privacy guarantees are expressed as $\rho$-zero-Concentrated Differential Privacy ($\rho$-zCDP), where similar results follow for other notions, such as approximate differential privacy (\Cref{lem:zcdp-eddp-conversion}).
We consider the number of inputs, $n$, fixed --- the neighboring relation is changing one vector among the inputs.

In the \emph{non-private} setting, the empirical mean $\tfrac{1}{n}\sum_j \vectorize{x}^{(j)}$ is known to yield the smallest $\ell_2$ error of size $O(\|\vectorize{\sigma}\|_2/\sqrt{n})$, which is also called the \emph{sampling error}. Since this sampling error is unavoidable, we are interested in the setting where the error introduced by differential privacy is larger than the sampling error.

\paragraph{Our results.}

We show, both theoretically and empirically, that a carefully chosen privacy budgeting can improve the $\ell_p$ error compared to existing methods when the vector \vectorize{\sigma} is skewed.
First, we provide a characterization of distributions $\mathcal{D}$ that are suitable for our budget allocation.
We say that $\mathcal{D}$ is \emph{$\vectorize{\sigma}$-well concentrated} if, roughly speaking, the norm of distance to the mean $\vectorize{\mu}$ of vectors sampled from $\mathcal{D}$ is unlikely to be much larger than vectors sampled from a multivariate exponential distribution with standard deviations given by $\vectorize{\sigma}$ or some root of $\vectorize{\sigma}$.
In the case $p=2$ our upper bound is particularly simple to state:
\begin{theorem}\label{thm:main} (simplified version)
	Suppose $\mathcal{D}$ is $\vectorize{\sigma}$-well concentrated and that we know $\vectorize{\hat{\sigma}}$ such that $\|\vectorize{\sigma} - \vectorize{\hat{\sigma}}\|_\infty < \|\vectorize{\sigma}\|_1/d$.
	Then for $n = \tilde{\Omega}\left(\max\left(\sqrt{d/\rho}, \rho^{-1}\right)\right)$ there is a mean estimation algorithm that satisfies $\rho$-zCDP and has expected $\ell_2$ error
	\[ \tilde{O}(1 + \|\vectorize{\sigma}\|_2/\sqrt{n} + \|\vectorize{\sigma}\|_1/(n\sqrt{\rho})), \]
	where $\tilde{O}$ suppresses polylogarithmic factors in $d$, $n$, and a bound on the $\ell_\infty$ norm of input vectors.
\end{theorem}

Comparing \Cref{thm:main} to applying the Gaussian mechanism directly provides a useful example.
Given the vectors $\vectorize{x}^{(1)},\dots,\vectorize{x}^{(n)}$ sampled from $\mathcal{N}(\vectorize{\mu}, \Sigma)$ for $\Sigma = \text{diag}(\vectorize{\sigma^2})$ where $\|\vectorize{\mu}\|_2 = O(\|\vectorize{\sigma}\|_2)$, and clipping at $C = \Theta(\|\vectorize{\sigma}\|_2)$, would give an $\ell_2$ error bound of
\[
\tilde{O}(\|\vectorize{\sigma}\|_2/\sqrt{n} + \sqrt{d} \|\vectorize{\sigma}\|_2/(n\sqrt{\rho})) \enspace .
\]
Crucially, \Cref{thm:main} is never worse, but can be better than applying the Gaussian mechanism directly.
Instance-optimal mean estimation~\citep{huang_instance-optimal_2021} matches the bound achieved by the Gaussian mechanism in this case since the diameter of the dataset is $\Theta(\|\vectorize{\sigma}\|_2)$.
The value of $\|\vectorize{\sigma}\|_1$ is in the interval $[\|\vectorize{\sigma}\|_2, \sqrt{d} \|\vectorize{\sigma}\|_2]$, where a smaller value of $\|\vectorize{\sigma}\|_1$ indicates larger skew.
Thus, \Cref{thm:main} is strongest when \vectorize{\sigma} has a skewed distribution such that $\|\vectorize{\sigma}\|_1$ is close to the lower bound of $\|\vectorize{\sigma}\|_2$, assuming that the privacy parameter $\rho$ is small enough that the error due to privacy dominates the sampling error.

When no estimate on the variance is known, \algorithmname still achieves compelling worst-case guarantees on the $\ell_2$ error. Theorem~\ref{thm:diameter:bound} showcases that skipping the scaling step, we recover the worst-case guarantees of instance-optimal mean estimating~\cite{huang_instance-optimal_2021}, without the need for an expensive random rotation of the input data. 
Our experimental evaluation in Section~\ref{sec:experiments} provides further support that even rough estimates on the variances suffice to achieve an advantage over state-of-the-art competitors, even if data is correlated. 
This makes \algorithmname a drop-in replacement for applications using differentially private statistical mean estimation.

%% file: tex/background.tex
\section{Preliminaries}\label{sec:background}
We provide privacy guarantees via \zcdplong in the \emph{bounded setting}, where the dataset has a fixed size, as defined in this section.
Since previous work measure error using different distance measures, we give both the definition for $\ell_p$ error (which we use), as well as Mahalanobis distance.
Lastly, we give an overview of private quantile estimation which is a central building block of our algorithm.

\subsection{Differential privacy}\label{sec:dp}
Differential privacy~\cite{dwork_calibrating_2006} is a statistical property of an algorithm that limits information leakage by introducing controlled randomness to the algorithm.
Formally, differential privacy restricts how much the output distributions can differ between any neighboring datasets.
We say that a pair of datasets are neighboring, denoted $x \sim x'$, if and only if there exists an $j \in [n]$ such that $\vectorize{x}^{(i)} = \vectorize{x}'^{(i)}$ for all $i \neq j$.

\chtodo{Do we need the definition of approximate DP? I don't think we use it anywhere.}

\begin{definition}[\cite{dwork_calibrating_2006} ($\varepsilon$, $\delta$)-Differential Privacy]\label{def:approximate-dp}
A randomized mechanism $\mathcal{M}$ satisfies ($\varepsilon$, $\delta$)-DP if and only if for all pairs of neighboring datasets $x \sim x'$ and all set of outputs $Z$ we have
$$\Pr[\mathcal{M}(x) \in Z] \leq e^\varepsilon \Pr[\mathcal{M}(x') \in Z] + \delta$$
\end{definition}

\begin{definition}[\cite{bun_concentrated_2016} \zcdplong]\label{def:zcdp}
Let $\mathcal{M}$ denote a randomized mechanism satisfying $\rho$-zCDP for any $\rho > 0$. Then for all $\alpha > 1$ and all pairs of neighboring datasets $x \sim x'$ we have
$$D_\alpha\left(\mathcal{M}(x) \vert \vert \mathcal{M}(x')\right) \leq \rho \alpha ,$$
where $D_\alpha(X \vert \vert Y)$ denotes the $\alpha$-Rényi divergence between two distributions $X$ and $Y$.
\end{definition}

\begin{lemma}[\cite{bun_concentrated_2016} \zcdp to $(\varepsilon, \delta)$-DP conversion]\label{lem:zcdp-eddp-conversion}
    If $\mathcal{M}$ satisfies $\rho$-zCDP, then $\mathcal{M}$ is $(\varepsilon, \delta)$-DP for any $\delta > 0$ and $\varepsilon=\rho + 2\sqrt{\rho \log(1/\delta)}$.
\end{lemma}

\begin{lemma}[\cite{bun_concentrated_2016} Composition]\label{lem:composition-zcdf}
    If $M_1$ and $M_2$ satisfy $\rho_1$-zCDP and $\rho_2$-zCDP, respectively.
    Then $M=(M_1, M_2)$ satisfies $(\rho_1 + \rho_2)$-zCDP.
\end{lemma}

The Gaussian Mechanism adds noise from a Gaussian distribution independently to each coordinate of a real-valued query output.
The scale of the noise depends on the privacy parameter and the $\ell_2$-sensitivity denoted $\Delta_2$.
A query $q$ has sensitivity $\Delta_2$ if for all $x \sim x'$ we have $\| q(x) - q(x') \|_2 \leq \Delta_2$.

\begin{lemma}[\cite{bun_concentrated_2016} The Gaussian Mechanism]\label{lem:gauss-zcdp}
    If $q:\mathbf{R}^{n \times d} \rightarrow \mathbf{R}^d$ is a query with sensitivity $\Delta_2$ then releasing $\mathcal{N}(q(x),\frac{\Delta_2^2}{2\rho}\mathbf{I})$ satisfies $\rho$-zCDP.
\end{lemma}

\subsection{Distance measures}\label{sec:distance-measures}

Here we introduce the distance measures for the error of a mean estimate.
Let $\vectorize{\tilde{\mu}} = \mathcal{M}(x)$ denote the mean estimate of a distribution with mean $\vectorize{\mu}$.
We measure the utility of our mechanism in terms of expected $\ell_p$ error as defined below.
Here $p$ is a parameter of our mechanism where the most common values for $p$ are $1$ (Manhattan distance), and $2$ (Euclidean distance).

\begin{definition}[$\ell_p$ error] \label{def:lp-error}
    For any real value $p \geq 1$ the $\ell_p$ error is
    \[
        \| \vectorize{\tilde{\mu}} - \vectorize{\mu} \|_p = \left(\sum_{i=1}^d \vert \tilde{\mu}_i - \mu_i \vert^p \right)^{1/p}
    \]
\end{definition}

Note that we sometimes present error guarantees in the form of the $p$th moment, i.e. $\mathbf{E}[\| \vectorize{\tilde{\mu}} - \vectorize{\mu} \|_p^p]$, when it follows naturally from the analysis.
Notice that the $p$th moment bounds the expected $\ell_p$ error for any $p \geq 1$ as $\mathbf{E}[\|\vectorize{\tilde{\mu}} - \vectorize{\mu}\|_p] \leq (\mathbf{E}[\|\vectorize{\tilde{\mu}} - \vectorize{\mu}\|_p^p])^{1/p}$.

As stated in Section~\ref{sec:related-work} several previous work on mean estimates measure error in Mahalanobis distance.
Although this is not the focus of our work we include the definition here for completeness.
A key difference between Mahalanobis distance measure and $\ell_p$ error is that the directions are weighted by the uncertainty of the underlying data.
We weight the error in all coordinates equally.

\begin{definition}[Mahalanobis distance]\label{def:mahalanobis}
    The error in Mahalanobis distance of a mean estimate
    for a distribution with covariance matrix $\Sigma$ is defined as
    \[
        \|\vectorize{\tilde{\mu}} - \vectorize{\mu}\|_{\Sigma} = \| \Sigma^{-1/2} (\vectorize{\tilde{\mu}}-\vectorize{\mu}) \|_2
    \]
\end{definition}

\subsection{Private quantile estimation}\label{sec:private-quantiles}

Our work builds on using differentially private quantiles, whose rank error (i.e. how many elements away from the desired quantile the output is) we will define for our utility analysis.
When estimating the $q$'th quantiles of a sequence $\vectorize{z}^{(1)} \dots \vectorize{z}^{(n)}$ with $\vectorize{z}^{(j)} \in \mathbf{R}^{d}$, we ideally return a vector $\vectorize{\tilde{z}} \in \mathbf{R}^{d}$ such that for each coordinate $i$ we have $\vert \{j \in [n] : \vectorize{z}_i^{(j)} \leq \vectorize{\tilde{z}}_i\}\vert / n = q$.
We use the notation $\privatequantile_{\rho}^M(\vectorize{z}^{(1)} \dots \vectorize{z}^{(n)}, q)$ to denote a $\rho$-zCDP mechanism that estimates the $q$'th quantiles of $\vectorize{z}^{(1)} \dots \vectorize{z}^{(n)}$ where each coordinate is bounded to $[-M,M]$.

There are multiple applicable choices for the instantiation of \privatequantile from the literature, e.g. \cite{smith_privacy-preserving_2011,huang_instance-optimal_2021, kaplan_differentially_2022}.
In this paper we will use the binary search based quantile \cite{smith_privacy-preserving_2011,huang_instance-optimal_2021} for our utility analysis (\Cref{sec:analysis}), and the exponential mechanism based quantile by \citet{kaplan_differentially_2022} in our empirical evaluation (\Cref{sec:experiments}).
Note that the choice of instantiation of \privatequantile has no effect on the privacy analysis of \algorithmname.
The binary search based method gives us cleaner theoretical results because the error guarantee of the exponential mechanism based technique is highly data-dependent, and the error is large for worst-case input.
Empirically, however, the exponential mechanism based quantile is more robust, which is why we use it in our experiments.

The binary search based quantile subroutine performs a binary search over the interval $[-M, M]$.
At each step of the search, branching is based on a quantile estimate for the midpoint of the current interval.
Since branching must be performed privately the privacy budget needs to be partitioned in advance, this limits the binary search to $T$ iterations.
\citet{huang_instance-optimal_2021} describes the algorithm for discrete input where $T$ is the base 2 logarithm of the size of the input domain.
Treating $T$ as a parameter of the algorithm gives us the following guarantees for $1$-dimensional input.

\begin{lemma}[Follows from~{\citet[Theorem~1]{huang_instance-optimal_2021}}]\label{lemma:privQuant-guarantees}
    \privatequantile satisfies $\rho'$-zCDP and with probability at least $1-\beta$ it returns an interval containing at least one point with rank error bounded by $\sqrt{T \log(T/\beta)/(2\rho')}$.
\end{lemma}

When estimating the quantiles of multiple dimensions, we split the privacy budget evenly across each dimension such that $\rho'=\rho/d$ for each invocation of \privatequantile.
By composition releasing all quantiles satisfies $\rho$-zCDP.

\begin{lemma}\label{lemma:approximate-quantile}
    With probability at least $1-\beta$, $\privatequantile_\rho^M$ returns a point that for all coordinates is within a distance of $M2^{-T}$ of a point with rank error at most $\sqrt{d T \log(T d /\beta)/2\rho}$ for the desired quantile.
\end{lemma}

\begin{proof}
    Running binary search for $T$ iterations splits up the range $[-M,M]$ in $2^T$ evenly sized intervals.
    By a union bound there is a point with claimed rank error in each of the intervals returned by \privatequantile.
    Returning the midpoint of each interval ensures we are at most distance $M2^{-T}$ from said point.
\end{proof}

Throughout this paper, we set $T=\log_2(M)$ unless otherwise specified such that the error distance of Lemma~\ref{lemma:approximate-quantile} is $1$.

%% file: tex/algorithm.tex
\section{Algorithm}\label{sec:algorithm}

\begin{figure*}[t]
	\centering
	\includegraphics[width=\linewidth]{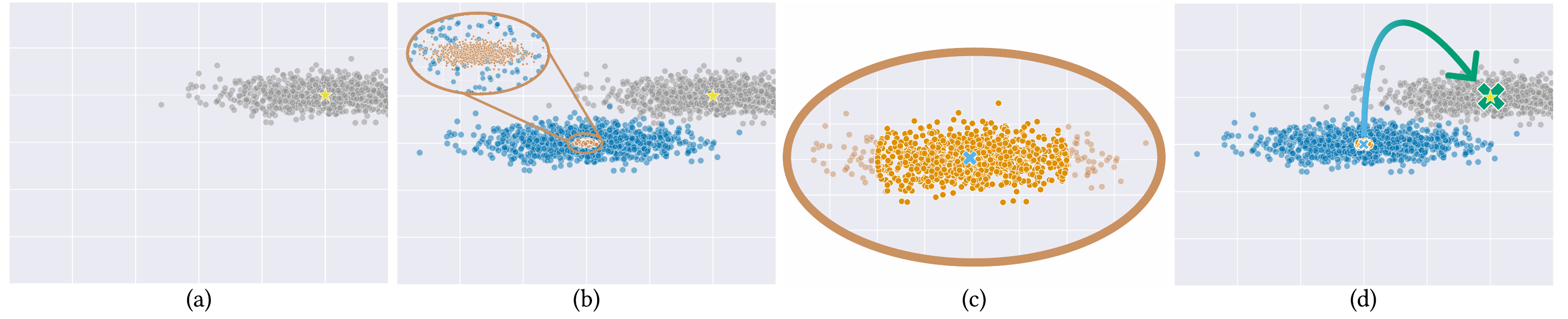}
	\caption{Step-by-step illustration of \algorithmname (\Cref{alg:our-algorithm}):
		 (a) Raw data, with statistical mean (yellow star),
		(b) Recentering (blue) and scaling (orange) corresponding to \Cref{alg:recenter-and-scale},
		(c) Clipping, as determined by \Cref{alg:estimate:clipping},
		(d) Private mean (green cross).
	}
	\label{fig:algorithm-illustrated}
\end{figure*}

Conceptually, \algorithmnamelong is a data-aware family of algorithms that exploits variance in the input data to tailor noise to the specific data at hand.
Since it is a family of algorithms, a \algorithmname needs to be instantiated for a given problem domain, i.e. for a specific $\ell_p$ error and data distribution.
In \Cref{sec:experiments} we showcase two such instantiations, using $\ell_1$ error for binary data, and using $\ell_2$ for Gaussian data.
Pseudocode for \algorithmname is shown in \Cref{alg:our-algorithm}, and a step-by-step illustration of \algorithmname is provided in \Cref{fig:algorithm-illustrated}.

\subsection{\algorithmname overview}\label{sec:plan-overview}

Based on the input data (\Cref{alg:input}, \Cref{fig:algorithm-illustrated}~(a)) \algorithmname first computes a private, rough, approximate mean $\vectorize{\tilde{\mu}}$ (\Cref{alg:center-prediction}).
\algorithmname then recenters the data around $\vectorize{\tilde{\mu}}$, and scales the data (\Cref{alg:recenter-and-scale}, \Cref{fig:algorithm-illustrated}~(b)) according to an estimate on the variance $\vectorize{\hat{\sigma}}^2$:
this scaling allows the privacy budget to be spent unevenly across dimensions.
Next, \algorithmname privately estimates a clipping threshold that contains most of the scaled data points (\Cref{alg:estimate:clipping}) 
and samples sufficient Gaussian noise for the privacy budget (\Cref{alg:draw-noise}).
Using the clipping threshold and the noise vector, \Cref{alg:plan-mean-estimation} in \algorithmname clips the scaled input  centered around~$\vectorize{\tilde{\mu}}$ (\Cref{fig:algorithm-illustrated}~(c)),
adds noise to the clipped data, and finally inverts the scaling and the recentering (\Cref{fig:algorithm-illustrated}~(d)). 
Each of these techniques has been used for private mean estimation before; we show that scaling data differently can lead to smaller $\ell_p$ error.

\begin{algorithm}[htb]
   \caption{\algorithmname}
   \label{alg:our-algorithm}
\begin{algorithmic}[1]
   \State \textbf{Input:} $\vectorize{x}^{(1)},\dots,\vectorize{x}^{(n)}\in \mathbf{R}^d$, estimate $\vectorize{\hat{\sigma}}^2 \in \mathbf{R}^d$  
   \label{alg:input}

   \State \begin{varwidth}[t]{\linewidth}
      \textbf{Parameters:} $M$ (\text{coordinate bound}), $p$ (\text{$\ell_p$ error norm}), \par
      \phantom{m}$\rho_1$ (\text{recentering budget}), $\rho_2$ (\text{threshold b.}), $\rho_3$ (\text{noise b.})
      \end{varwidth}
   
   \label{alg:parameters}

   \State $\hat{\Sigma} \gets \text{diag}(\vectorize{\hat{\sigma}}^2)$; $k \leftarrow \sqrt{n}+\Theta\left(\sqrt{1/\rho_2}\right)$
   \label{algo:set:parameters}

   \State $\vectorize{\tilde{\mu}} \leftarrow \privatequantile_{\rho_1}^M(\vectorize{x}^{(1)},\dots, \vectorize{x}^{(n)}, 1/2)$
   \label{alg:center-prediction}

   \State $\vectorize{y}^{(j)} \leftarrow (\vectorize{x}^{(j)} - \vectorize{\tilde{\mu}})\, \hat{\Sigma}^{-1/(p+2)}$ for $j=1,\dots,n$
   \label{alg:recenter-and-scale}

   \State $C \leftarrow \privatequantile_{\rho_2}^{M\sqrt{d}}(\|\vectorize{y}^{(1)}\|_2, \dots, \|\vectorize{y}^{(n)}\|_2, \tfrac{n-k}{n})$
   \label{alg:estimate:clipping}

   \State sample $\vectorize{\eta} \sim \mathcal{N}(\mathbf{0},\frac{2C^2}{\rho_3}\mathbf{I})$\label{alg:draw-noise}

   \State\textbf{return} $\vectorize{\tilde{\mu}}+\tfrac{1}{n} \left(\sum_j \min\left\{\frac{C}{\|\vectorize{y}^{(j)}\|_2}, 1\right\} \vectorize{y}^{(j)} + \vectorize{\eta}\right) \hat{\Sigma}^{1/(p+2)}$
   \label{alg:plan-mean-estimation}
\end{algorithmic}
\end{algorithm}

\subsection{\algorithmname building blocks}\label{sec:plan-details}

\chtodo{$\text{PrivQuantile}^M$ has been defined to take input from [-M,M]. But on line 5 input is always positive. It is not important but it does mean we are "wasting" one iteration if we run the algorithm as described.}
\matodo{Make it clear what the input/output of PrivQuantile is.}

Like \instanceoptimallong~\cite{huang_instance-optimal_2021}, our algorithm computes a rough estimate of the mean as a private estimate $\vectorize{\tilde{\mu}}\in\mathbf{R}^d$ of the coordinate-wise median.
This step uses the assumption that all coordinates are in the interval $[-M,M]$. %
It is known that a finite output domain is needed for private quantile selection to be possible~\cite{bun_differentially_2015}. %
Suppose $\ell_p$ is the error measure we are aiming to minimize.
The next step is to scale coordinates by multiplying with the diagonal matrix $\hat{\Sigma}^{-1/(p+2)}$, where $\hat{\Sigma} = \text{diag}(\vectorize{\hat{\sigma}}^2)$ is the diagonal matrix of variance estimates $\vectorize{\hat{\sigma}}^2$.
Note that since $\hat{\Sigma}$ is a diagonal matrix, it is not necessarily close to the covariance matrix of $\mathcal{D}$ outside of the diagonal. 

Note that for $p=2$ the exponent of $\hat{\Sigma}$ is $-1/4$, which is different from the exponent of $-1/2$ that would be used in order to compute an estimate with small Mahalanobis distance~(See e.g.~\cite{brown_covariance-aware_2021}).
In fact, our choice of scaling lies between two natural choices. 
If we removed scaling the magnitude of noise would be the same across all coordinates which corresponds to the standard Gaussian mechanism. 
This is problematic if many coordinates have small variance compared to others as we would still add a lot of noise to the low variance coordinates.
On the other hand, scaling by $\hat{\Sigma}^{-1/2}$ results in noise at each coordinate with magnitude linear proportional to $\hat{\sigma}_i$. 
That approach adds excessive noise to the high variance coordinates.
We strike a middle ground when scaling by $\hat{\Sigma}^{-1/4}$.
We add less noise than the first approach for low variance coordinates, and less noise than the second approach for high variance coordinates.
Our choice of scaling is one of the main contributions of our work. 
A key feature of the scaling is that we adjust to the error measure.
We discuss scaling further in \Cref*{app:calibration} where we show that our choice of scaling is optimal for mean estimation in a simpler setting. 

The scaling stretches the $i$th coordinate by a factor $\hat{\sigma}_i^{-2/(p+2)}$.
Since $\hat{\sigma}^2_i \approx \sigma^2_i$ this changes the standard deviation on the $i$th coordinate from $\sigma_i$ to roughly $\sigma_i^{p/(p+2)}$.
Next, we compute vectors $\vectorize{y}^{(j)}$ that represent the (stretched) differences $\vectorize{x}^{(j)}-\vectorize{\tilde{\mu}}$.
Conceptually, we now want to estimate the mean of $\vectorize{y}^{(1)},\dots,\vectorize{y}^{(n)}$, which in turn implies an estimate of $(\tfrac{1}{n}\sum_j \vectorize{x}^{(j)}) - \vectorize{\tilde{\mu}}$.
The mean of $\vectorize{y}^{(1)},\dots,\vectorize{y}^{(n)}$ is estimated using the Gaussian mechanism.
In order to find a suitable scaling of the noise we privately find a quantile $C$ of the lengths $\|\vectorize{y}^{(1)}\|_2,\dots,\|\vectorize{y}^{(n)}\|_2$ (all shorter than $M\sqrt{d}$ by assumption)
such that approximately $k$ vectors have length larger than $C$, and clip vectors to length at most $C$.
Clipping, private mean estimation, scaling, and adding back $\vectorize{\tilde{\mu}}$ are all condensed in the estimator in \Cref{alg:plan-mean-estimation} of \Cref{alg:our-algorithm}.

%% file: tex/analysis.tex
\section{Analysis}\label{sec:analysis}

We will show that \algorithmname returns a private mean estimate that---assuming $\mathcal{D}$ is well behaved---has small expected $\ell_p$ error with probability at least $1-\beta$.
All probabilities are over the joint distribution of the input samples and the randomness of the \algorithmname mechanism.
For simplicity we focus on the case $p=2$, but the analysis extends to any $p\geq 1$ as we will show at the end of this section.

We make the following assumptions:
\begin{itemize}
    \item We assume that for a known parameter $M$, all inputs are in $\vectorize{x}^{(j)}\in [-M, M]^d$ for $j=1,\dots,n$. (If this is not the case, the algorithm will clip inputs to this cube, introducing additional clipping error.)
	Also, we assume that data has been scaled sufficiently such that $\sigma_i \geq 1$ for $i=1,\dots,d$. (We discuss what to do if this is not the case in Section~\ref{sec:low-variance}). 
    \item We are given a vector $\vectorize{\hat{\sigma}}\in\mathbf{R}^d$ such that
	\begin{align}\label{assumption:variance-estimates}
		\sigma_i \leq \hat{\sigma}_i <  \sigma_i + \|\vectorize{\sigma}\|_1/d \enspace .
	\end{align}
	If no such vector is known we will have to compute it, spending part of the privacy budget. We consider this a separate question and evaluate possible techniques in Section~\ref{sec:experiments}.
\end{itemize}

\begin{definition}
	\label{def:well-concentrated}
	Consider a distribution $\mathcal{D}$ over $\mathbf{R}^d$, denote the mean and standard deviation of the $i$th coordinate by $\mu_i$ and $\sigma_i$, respectively.
	We say that $\mathcal{D}$ is \emph{$\vectorize{\sigma}$-well concentrated} if for any vector $\vectorize{\hat{\sigma}}\in \mathbf{R}^d$ with	$\sigma_i \leq \hat{\sigma}_i <  \sigma_i + \|\vectorize{\sigma}\|_1/d$, the following holds for $t > \ln d$
	\begin{align}
    & \Pr_{X \sim \mathcal{D}}\left[ \sum_{i=1}^d (X_i - \mu_i)^2 > t \|\vectorize{\sigma}\|_2^2 \right] = \exp(-{\Omega}(t))\label{assumption:concentrated} \\
	& \Pr_{X \sim \mathcal{D}}\left[ \sum_{i=1}^d (X_i - \mu_i)^2/\hat{\sigma}_i > t \|\vectorize{\sigma}\|_1 \right] = \exp(-{\Omega}(t)) \enspace \label{assumption:scaled_concentrated}
	\end{align}
\end{definition}
Intuitively, these assumptions require concentration of measure of the norms before and after scaling.

\subsection*{Analysis outline}\label{sec:analysis-outline}

\paragraph{Privacy.}
It is not hard to see that \algorithmname satisfies $\rho$-zCDP with $\rho = \rho_1+\rho_2+\rho_3$:
The computation of $\vectorize{\tilde{\mu}}$ satisfies $\rho_1$-zCDP and the computation of $C$ satisfies $\rho_2$-zCDP by definition of \privatequantile.
Finally, given the values $C$ and $\vectorize{\tilde{\mu}}$ the $\ell_2$-sensitivity of $\sum_j \min\left\{\frac{C}{\|\vectorize{y}^{(j)}\|_2}, 1\right\} \vectorize{y}^{(j)}$ with respect to an input $\vectorize{x}^{(j)}$ is~$2C$, so adding Gaussian noise with variance $2C^2/\rho_3$ gives a $\rho_3$-zCDP mean estimate.
By composition, and since the returned estimator is a post-processing of these private values, the estimator is $\rho$-zCDP with $\rho = \rho_1+\rho_2+\rho_3$.

\paragraph{Utility.}
\cite{huang_instance-optimal_2021,kamath_privately_2019} present multiple lower bounds that show that there exist inputs such that mean estimation in $\mathbf{R}^d$ requires $n=\Omega(\sqrt{d/\rho})$ samples to achieve meaningful utility.
In our analysis, we need the same assumption and require $n = \tilde{\Omega}(\sqrt{d/\rho})$, where the $\tilde{\Omega}$ notation hides polylogarithmic factors in $d$, $\beta$, and $M$.
Let $\hat{\Sigma}$ denote the scaling matrix $\text{diag}(\vectorize{\hat{\sigma}}^2)$.
In the following, we assume that $\rho_1,\rho_2,\rho_3$ are fixed fractions of $\rho$.

The utility analysis has three parts:
\begin{enumerate}
	\item First, we argue that $|\tilde{\mu}_i - \mu_i| < 3\sigma_i$ for all $i=1,\dots,d$ with high probability.
	\item Let $J$ denote the set of indices for which $\|\vectorize{y}^{(j)}\|_2 > C$.
	We argue that with high probability,
	\[\sum_{j\in J} \| \vectorize{x}^{(j)} - \vectorize{\tilde{\mu}} \|_2 = \tilde{O}\left(\left(k + \sqrt{\tfrac{1}{\rho}}\right) \|\vectorize{\sigma}\|_2\right),\]
	bounding the error due to clipping.
	\item Finally, we argue that with high probability
	\[\left\| \vectorize{\eta}\, \hat{\Sigma}^{\tfrac{1}{p+2}} \right\|^2_2 = \tilde{O}(\|\vectorize{\sigma}\|_1^2 / \rho),\]
	bounding the error due to Gaussian noise.
\end{enumerate}

\subsection{Part 1: Bounding $|\vectorize{\tilde{\mu}} - \vectorize{\mu}|$}\label{sec:analysis-bounding-mu}

In the following $K$ is a universal constant chosen to be sufficiently large.

\begin{lemma}\label{lemma:coordinatewise:median}
    Assuming $n > K \max ( \sqrt{d \log(M) \log( \log(M) d/\beta) / \rho},\allowbreak \log(d/\beta) )$ then with probability $1-\beta$, the coordinate-wise median ${\tilde{\mu}_i}$ satisfies ${\tilde{\mu}}_i \in [{\mu}_i - 3{\sigma}_i, {\mu}_i + 3{\sigma}_i]$ for all $i=1,\dots,d$.
\end{lemma}
\begin{proof}
	Sample $X=(X_1, \ldots, X_d) \sim \mathcal{D}$. By Chebychev's inequality, for each $i$ with $1 \leq i \leq d$:
	\[\Pr[|X_i- {\mu}_i| \leq 2{\sigma}_i] \geq 1 - \frac{{\sigma}_i^2}{(2 {\sigma}_i)^2} = 3/4 \enspace . \]
	Fixing $i$, these events are independent for $j=1,\dots,n$ so by a Chernoff bound with probability $1-\exp(-\Omega(n))$ there are at least $\tfrac{2}{3} n$ indices $j$ such that $|X^{(j)}_i-{\mu}_i| \leq 2{\sigma}_i$.
	Condition on the event that this is true for $i=1,\dots,d$.

	If the rank error of the median estimate ${\tilde{\mu}}_i$ is smaller than $n/6$, ${\tilde{\mu}}_i \in [{\mu}_i - 3{\sigma}_i, {\mu}_i + 3{\sigma}_i]$, using the assumption that ${\sigma}_i \geq 1$.
	By Lemma~\ref{lemma:approximate-quantile} this is true for every $i$ with probability at least $1-\beta$, given our assumption on $n$ and $K$ if $K$ is a sufficiently large constant.
\end{proof}

\subsection{Part 2: Bounding clipping error}\label{sec:clipping:error}

Let $J = \{ j\in \{1,\dots,n\} \; | \; \|\vectorize{y}^{(j)}\|_2 > C \}$ denote the set of indices of vectors affected by clipping in \algorithmname.
By the triangle inequality and assuming the bound of part 1 holds,
\begin{align*}
	\sum_{j\in J} \left\|(\vectorize{x}^{(j)} - \vectorize{\tilde{\mu}}) \right\|_2
	&\leq \sum_{j\in J} \left\|(\vectorize{x}^{(j)} -\vectorize{\mu}) \right\|_2 + |J|\cdot\left\|\vectorize{\mu} - \vectorize{\tilde{\mu}} \right\|_2\\
	&\leq \sum_{j\in J} \left\|(\vectorize{x}^{(j)} -\vectorize{\mu}) \right\|_2 + 3 |J|\cdot \|\vectorize{\sigma}\|_2 \enspace .
\end{align*}

By assumption~(\ref{assumption:concentrated}) the probability that $\| \vectorize{x}^{(j)} - \vectorize{\mu} \|_2^2 > t \|\vectorize{\sigma}\|_2^2$ is exponentially decreasing in $\Omega(t)$ for $t \geq \ln d$, so setting $t  = \Omega\left(\max\{\ln d, \ln(n/\beta)\}\right)$ 
we have $\| \vectorize{x}^{(j)} - \vectorize{\mu} \|_2 = \tilde{O}(\|\vectorize{\sigma}\|_2)$ for all $j=1,\dots,n$ with probability at least $1-\beta$.
Using the triangle inequality again we can now bound
\begin{align*}
\sum_{j\in J} \| \vectorize{x}^{(j)} - \vectorize{\mu} \|_2 = \tilde{O}(|J|\cdot\|\vectorize{\sigma}\|_2) \enspace .
\end{align*}
For this part of the analysis, we assume that \privatequantile is instantiated as in Lemma~\ref{lemma:approximate-quantile}, but we return the maximum of the interval instead of the midpoint.
This ensures that no vectors whose length fall in the returned interval are clipped.
This choice increases the value of $C$ by $1$ which slightly increases the magnitude of noise in part 3 of the analysis.

With probability at least $1-\beta$ (over the random choices of the private quantile selection algorithm) there are at most $\tilde{O}\left(k+\sqrt{\tfrac{1}{\rho}}\right)$ vectors $\vectorize{y}^{(j)}$ for which $\|\vectorize{y}^{(j)}\|_2 > C$ (assuming $T=\log_2(M\sqrt{d})$), cf. Lemma~\ref{lemma:approximate-quantile} and using that $\left\|\vectorize{y}^{(j)}\right\|_2 \leq M\sqrt{d}$.
That is, we have $|J| = \tilde{O}\left(k + \sqrt{\tfrac{1}{\rho}}\right)$, and using the bounds above we get
\begin{align*}
	\sum_{j\in J} \| \vectorize{x}^{(j)} - \vectorize{\tilde{\mu}} \|_2 = \tilde{O}(|J|\cdot \|\vectorize{\sigma}\|_2) = \tilde{O}\left(\left(k + \sqrt{\tfrac{1}{\rho}}\right) \|\vectorize{\sigma}\|_2\right) \enspace .
\end{align*}

\subsection{Part 3: Bounding noise}\label{sec:noise:error}

By triangle inequality:
\begin{align*}
\|\vectorize{y}^{(j)}\|_2
& = \left\| (\vectorize{x}^{(j)} - \vectorize{\tilde{\mu}})\, \hat{\Sigma}^{-\tfrac{1}{p+2}} \right\|_2 \\
& \leq \left\| (\vectorize{x}^{(j)} - \vectorize{\mu})\, \hat{\Sigma}^{-\tfrac{1}{p+2}} \right\|_2 + \left\| (\vectorize{\mu} - \vectorize{\tilde{\mu}})\, \hat{\Sigma}^{-\tfrac{1}{p+2}} \right\|_2
\end{align*}

For $p=2$, the $i$th coordinate of $(\vectorize{x}^{(j)} - \vectorize{\mu})\, \hat{\Sigma}^{-1/(p+2)}$ equals $({x}^{(j)}_i - {\mu}_i)/\sqrt{{\hat{\sigma}}_i}$, so assumption (\ref{assumption:scaled_concentrated}) implies that the length of the first vector is $\tilde{O}(\sqrt{\|\vectorize{\sigma}\|_1})$ with probability at least $1-\beta/2$.
From part 1 we know (again for $p=2$) that the $i$th coordinate of $(\vectorize{\mu} - \vectorize{\tilde{\mu}})\, \hat{\Sigma}^{-1/(p+2)}$ has absolute value at most
\[ 3{\sigma}_i/\sqrt{{\hat{\sigma}}_i} \leq 3\sqrt{{\sigma}_i}, \]
where the inequality uses the lower bound on ${\hat{\sigma}}_i$ in assumption (\ref{assumption:variance-estimates}).
So $\|(\vectorize{\mu} - \vectorize{\tilde{\mu}})\, \hat{\Sigma}^{-1/(p+2)}\|_2 \leq 3\sqrt{\|\vectorize{\sigma}\|_1}$, and we get that $\|\vectorize{y}^{(j)}\|_2 = \tilde{O}(\sqrt{\|\vectorize{\sigma}\|_1})$ for all $j$ with probability at least $1-\beta/2$.

The value of $C$ is bounded by the maximum length of a vector $\|\vectorize{y}^{(j)}\|_2 + 1$ (accounting for the interval size of \privatequantile in part 2), and thus by a union bound, with probability at least $1-\beta$,  $C^2 = \tilde{O}(\|\vectorize{\sigma}\|_1)$.
The scaled noise vector $\vectorize{\eta}\, \hat{\Sigma}^{1/(p+2)}$ has distribution $\mathcal{N}(\vectorize{0},\frac{2C^2}{\rho_3}\hat{\Sigma}^{1/(p+2)})$, so using ${\hat{\sigma}}_i \leq {\sigma}_i + \|\vectorize{\sigma}\|_1/d$ from assumption (\ref{assumption:variance-estimates}), for $p=2$:
\begin{align*}
	\mathbf{E}[\|\vectorize{\eta}\, \hat{\Sigma}^{1/(p+2)}\|_2^2]
	& = \tfrac{2C^2}{\rho_3} \sum_{i=1}^d \hat{\Sigma}_{ii}^{2/(p+2)}
	= \tfrac{2C^2}{\rho_3} \sum_{i=1}^d \hat{{\sigma}}_i\\
	& \leq \tfrac{2C^2}{\rho_3} \sum_{i=1}^d ({\sigma}_i + \|\vectorize{\sigma}\|_1/d)
	= \tilde{O}(\|\vectorize{\sigma}\|_1^2 / \rho_3) \enspace .
\end{align*}

\subsection{Proof of Theorem~\ref{thm:main}}\label{sec:analysis-proof}

The output of \algorithmname can be written as
\begin{align*}
\tfrac{1}{n} \sum_{j=1}^n \vectorize{x}^{(j)} - \tfrac{1}{n} \sum_{j\in J} \tfrac{\|\vectorize{y}^{(j)}\|_2-C}{\|\vectorize{y}^{(j)}\|_2}\, (\vectorize{x}^{(j)}-\vectorize{\tilde{\mu}}) + \tfrac{1}{n} \vectorize{\eta} \, \hat{\Sigma}^{1/(p+2)}
\end{align*}

Using that $\tfrac{\|\vectorize{y}^{(j)}\|_2-C}{\|\vectorize{y}^{(j)}\|_2} < 1$ for $j\in J$ and the triangle inequality, the $\ell_2$ estimation error of \algorithmname can thus be bounded as
\begin{align*}
\left\|\tfrac{1}{n} \sum_{j=1}^n \vectorize{x}^{(j)} - \vectorize{\mu}\right\|_2 + \tfrac{1}{n} \sum_{j\in J} \|\vectorize{x}^{(j)}-\vectorize{\tilde{\mu}}\|_2 + \tfrac{1}{n} \|\vectorize{\eta} \, \hat{\Sigma}^{1/(p+2)}\|_2
\end{align*}
The first term is the sampling error, which is $\tilde{O}(\|\vectorize{\sigma}\|_2/\sqrt{n})$ with probability at least $1-\beta$.
The sum over $J$ was shown in part 2 to be $\tilde{O}\left(\left(k + \sqrt{\tfrac{1}{\rho}}\right)\|\vectorize{\sigma}\|_2\right)$, so for $k\leq \sqrt{n}$ and using the assumption $n = \tilde{\Omega}(\max(\sqrt{d/\rho}, \rho^{-1}))$, this term is bounded by $\tilde{O}(\|\vectorize{\sigma}\|_2/\sqrt{n})$.
Finally, the noise term in part~3 is $\tilde{O}(\|\vectorize{\sigma}\|_1/(n\sqrt{\rho}))$.

We are now ready to state a more detailed version of Theorem~\ref{thm:main}.
\begin{theorem}\label{thm:main-detailed}
	For sufficiently large $n = \tilde{\Omega}(\sqrt{d/\rho})$, \algorithmname with parameters $k=\sqrt{n}$ and $\rho_1 = \rho_2 = \rho_3 = \rho/3$ is $\rho$-zCDP.
	If inputs are independently sampled from a $\vectorize{\sigma}$-well concentrated distribution, the mean estimate has expected $\ell_2$ error $$\tilde{O}(1 + ||\vectorize{\sigma}||_2/\sqrt{n} + ||\vectorize{\sigma}||_1/(n\sqrt{\rho})),$$  where $\tilde{O}$ suppresses polylogarithmic dependencies on $n$, $d$, and the bound $M$ on the $\ell_\infty$ norm of inputs.
\end{theorem}

\begin{proof}
	The discussion above argued for an expected $\ell_2$ error in the absence of ``failure events.'' 
	The probability that no such event happens is at least $1-\beta$. 
	The largest possible error is the diameter of the cube from which the input vectors come, $2M\sqrt{d}$.
	Thus, the largest contribution to the expected error from a probability $\beta$ event is $2 M\sqrt{d} \beta$.
	Choosing $\beta = \frac{1}{M\sqrt{d}}$ guarantees that the expected contribution from these events is $O(1)$.
\end{proof}

\subsection{General $\ell_p$ error}\label{sec:analysis-generalization}

To address the general case we need a definition of well-concentrated that depends on $p$.
For simplicity we assume that $p$ is a positive integer constant.
\begin{definition}
	\label{def:well:concentrated:generalized}
	For integer constant $p\geq 1$ consider a distribution $\mathcal{D}$ over $\mathbf{R}^d$, where the $i$th coordinate has mean ${\mu}_i$ and $p$th central moment ${\sigma}_i^p$.
	We say that $\mathcal{D}$ is \emph{$(\vectorize{\sigma},p)$-well concentrated} if for any vector $\vectorize{\hat{\sigma}}\in \mathbf{R}^d$ with	${\sigma}_i \leq {\hat{\sigma}}_i <  \left({\sigma}_i^{\frac{2p}{p +2 }} + \|\vectorize{\sigma}^{\frac{2p}{p + 2}}\|_1/d\right)^{\frac{p + 2}{2p}}$, the following holds for $t > \ln d$:
	\begin{align}
    & \Pr_{X \sim \mathcal{D}}\left[ \sum_{i=1}^d |X_i - {\mu}_i|^p > t \|\vectorize{\sigma}\|_p^p \right] = \exp(-{\Omega}(t))\label{assumption:concentrated_generalized} \\
	& \Pr_{X \sim \mathcal{D}}\left[ \sum_{i=1}^d \frac{(X_i - {\mu}_i)^2}{{\hat{\sigma}}_i^{4/(p+2)}} > t \|\vectorize{\sigma}^{p/(p+2)}\|_2^2 \right] = \exp(-{\Omega}(t)) \enspace . \label{assumption:scaled_concentrated_generalized}
	\end{align}
\end{definition}

\begin{theorem}\label{thm:main-detailed-generalized}
	For sufficiently large $n = \tilde{\Omega}(\max(\sqrt{d/\rho}, \rho^{-1}))$, \algorithmname with parameters $k=\sqrt{n}$ and $\rho_1 = \rho_2 = \rho_3 = \rho/3$ is $\rho$-zCDP.
	If inputs are independently sampled from a $(\vectorize{\sigma},p)$-well concentrated distribution, the mean estimate has expected $\ell_p$ error
	\[
	\tilde{O}\left(1 + \frac{||\vectorize{\sigma}||_p}{\sqrt{n}} + \frac{\|\vectorize{\sigma}\|_{2p/(p+2)}}{n\sqrt{\rho}}\right),
	\]
	where $\tilde{O}$ suppresses polylogarithmic dependencies on $1/\beta$, $n$, $d$, and the bound $M$ on the $\ell_\infty$ norm of inputs.
\end{theorem}

The proof of this theorem follows along the lines of the proof in this section and can be found in Appendix~\ref{app:generalized:lp}.
In comparison, the standard Gaussian mechanism for $\ell_2$ sensitivity $\| \vectorize{\sigma} \|_2$ has expected $\ell_p$ error $\tfrac{d^{1/p}\| \vectorize{\sigma} \|_2}{\sqrt{\rho}}$.

\input{tex/low-variance.tex}

\section{Examples of well-concentrated distributions}\label{sec:well-concentrated}
In the following we give examples for some $(\vectorize{\sigma}, p)$-well concentrated distributions.
We start with a general result.

\begin{lemma}
	\label{lemma:example:well:concentrated}
	For integer constant $p \geq 1$, consider a distribution $\mathcal{D}$ over $\mathbf{R}^d$ where the $i$th coordinate is independently distributed with mean ${\mu}_i$ and standard deviation ${\sigma}_i$ such that for $X \sim \mathcal{D}$, the $p$th moment $\mathbf{E}\left[\vert X_i-{\mu}_i \vert^{p}\right] \leq K {\sigma}_i^{p}$ for some constant $K$.
	Then $\mathcal{D}$ is $(\vectorize{\sigma}, p)$-well concentrated.
\end{lemma}

\begin{proof}
	We first prove the bound ~\eqref{assumption:concentrated_generalized} from Definition~\ref{def:well:concentrated:generalized}.
	Sample $X$ from $\mathcal{D}$ and define $Y_i = \vert X_i - {\mu}_i \vert ^p - \mathbf{E}[\vert X_i - {\mu}_i \vert ^p]$ and note that $\sum Y_i \leq \sum \left| X_i - {\mu}_i\right|^p$.
	Each $Y_i$ is zero-centered, so we may apply Bernstein's inequality (Lemma~\ref{lem:bernstein})
	\begin{align*}
		\Pr\left[ \sum_{i=1}^d Y_i > t \|\vectorize{\sigma}\|_p^p \right] \leq \exp\left(-\frac{t^2 \|\vectorize{\sigma}\|_p^{2p}/2}{\sum_{i = 1}^d \mathbf{E}[Y_i^2] + M t \|\vectorize{\sigma}\|_p^p / 3 }\right).
	\end{align*}
	We distinguish two cases depending on which term is dominating the denominator.
	In the first case,
	$$\sum_{i = 1}^d \mathbf{E}\left[Y_i^2\right] \leq \sum_{i = 1}^d \mathbf{E}\left[\left(X_i - {\mu}_i\right)^{2p}\right] \leq K \|\vectorize{\sigma}^{2p}\|_1 \leq \|\vectorize{\sigma}\|_p^{2p},$$
	where we applied the triangle inequality.
	This means that the probability of $\sum |X_i - {\mu}_i|^p$ to exceed $t \|\vectorize{\sigma}\|_p^p$ is $\exp\left(-\Omega\left(t^2\right)\right)$.
	In the second case,
	$$\frac{t^2 \|\vectorize{\sigma}\|_p^{2p}/2}{M t \|\vectorize{\sigma}\|_p^p/ 3} = \frac{3 t\|\vectorize{\sigma}\|_p^{p}}{2M}.$$
	By Chebychev's inequality almost surely $Y_i \leq C \max_{i}{\sigma}_i^p \leq \|{\sigma}\|_p^{p}$. Thus, the probability is bounded by $\exp(-\Omega(t))$ in this case.

	For the second property~\eqref{assumption:scaled_concentrated_generalized}, note that $\sum (X_i-{\mu}_i)^{2}/{\hat{\sigma}}_i^{4/(p + 2)}$ is maximized for ${\hat{\sigma}}_i = {\sigma}_i$.
	For this choice, the calculations are analogous to the ones above.
\end{proof}

We will now proceed with studying the Gaussian distribution for $\ell_2$ error and a sum of Poisson trials for $\ell_1$ error.
These distributions are the basis of the empirical evaluation in Section~\ref{sec:experiments}.

\subsection{Gaussian data}

\begin{lemma}
	Let $p = 2$ and fix $\vectorize{\mu} \in \mathbf{R}^d$ and $\Sigma \in \mathbf{R}^{d \times d}$. The multivariate normal distribution $\mathcal{N}\left(\vectorize{\mu}, \Sigma\right)$ is ($\vectorize{\sigma}$, 2)-well concentrated.
\end{lemma}

Before presenting the proof, we remark that the independent case with diagonal covariance matrix $\Sigma$ can easily be handled by Chernoff-type bounds.
Furthermore, Lemma~\ref{lemma:example:well:concentrated} holds for $\mathcal{N}\left(\vectorize{\mu}, \Sigma\right)$ for diagonal covariance matrix.
In the lemma, we handle the general case of (non-)diagonal covariance matrices, allowing for dependence among the coordinates.

\begin{proof}
	Consider property~\eqref{assumption:concentrated} of Definition~\ref{def:well-concentrated} and fix any $ i \in \{1, \ldots, d\}$.
	Let ${\sigma}_i^2 = \Sigma_{ii}$ and sample $X_i \sim \mathcal{N}({\mu}_i, {\sigma}_i^2)$.
	$X_i$ is normally distributed, so we can make use of \cite[Theorem 9.3]{mitzenmacher:upfal} which states that for all $t > 0$ it holds that  $$\Pr[|X - {\mu}_i| \geq t] \leq 2\text{exp}(-t^2/(2 {\sigma}_i^2)).$$
	Fix a value $t > 0$ and let $t' = (t + \ln(d)){\sigma}_i^2$.
	We proceed to bound the deviation as follows:
	\begin{align*}
		\Pr[(X_i - {\mu}_i)^2 \geq t'] &= \Pr[|X_i - {\mu}_i| \geq \sqrt{t'}]\\
		&\leq 2\exp\left(-\left((t  + \ln(d)) {\sigma}_i^2\right)/\left(2 {\sigma}_i^2\right)\right) \\
		&= 1/d \exp(-{\Omega}(t)).
	\end{align*}
	Consider the case that we sample $X_1, \ldots, X_d$ from $\mathcal{N}\left(\vectorize{\mu}, \Sigma\right)$. By a union bound over $X_1, \ldots, X_d$, we may assume that with probability at least $1 - \text{exp}\left(-{\Omega}(t)\right)$, for all $i \in \{1, \ldots, d\}$ we have $(X_i - {\mu}_i)^2 \leq (t + \ln(d)){\sigma}_i^2$, which implies that their sum is at most $(t  + \ln(d)) \|\vectorize{\sigma}\|^2_2 = \tilde{O}(t \|\vectorize{\sigma}\|^2_2)$.

	Next, consider property~\eqref{assumption:scaled_concentrated} of Definition~\ref{def:well-concentrated}.
	Sample $X \sim \mathcal{N}({\mu}_i, {\sigma}_i^2)$ and consider the transformation $(X - {\mu}_i)^2/{\sigma}_i$.
	Setting $t' = (t  + \ln(d)) {\sigma}_i$, the same calculations as above show that with probability at most $1/d \exp(-{\Omega}(t))$, $(X - {\mu}_i)^2/{\sigma}_i$ is larger than $(t + \ln(d)){\sigma}_i$.

	Consider the case that we sample $X$ from $\mathcal{N}\left(\vectorize{\mu}, \Sigma\right)$. %
	Again using a union bound, with probability at least $1 - \exp(-{\Omega}(t))$, all scaled values are within $(t + \ln(d))\vectorize{\sigma}_i$.
	Conditioned on this, $\sum_{i=1}^d (X_i - {\mu}_i)^2/{\hat{\sigma}}_i \leq (t +  \ln(d)) \|\vectorize{\sigma}\|_1$, which finishes the proof.
\end{proof}

\subsection{Binary data}

Consider length-$d$ binary strings such that bit $i$ is set independently and at random with probability $\probability_i$.
Lemma~\ref{lemma:example:well:concentrated} is not applicable because the $p$th moment is roughly equal to $\max\{\probability_i, 1-\probability_i\}$.

\begin{lemma}
	\label{lem:binary:concentrated}
	Let $d \geq 1$ be an integer, and let $\mathcal{D}_\text{binary}$ be the distribution over length-$d$ binary strings such that the bit in position $i$ is set with probability $\probability_i$. If for each $i \in \{1, \ldots, d\}$, $\sigma_i^2 = \probability_i (1-\probability_i) \geq 1/d^{2/5}$, then $\mathcal{D}_\text{binary}$ is $(\vectorize{\sigma}, 1)$-well concentrated.
\end{lemma}

\begin{proof}
	Since we assumed in Section~\ref{sec:analysis} that $\sigma_i \geq 1$, we consider the mapping $x \mapsto d^{1/5}x =: \bar{x}$.
	Consider Property~\eqref{assumption:concentrated_generalized} of Definition~\ref{def:well:concentrated:generalized}.
	Since ${\bar{\sigma}}_i^2 \geq 1$, $\|\vectorize{\bar{\sigma}}\|_1 \geq d$.
	Using a generalized Chernoff bound (Lemma~\ref{lemma:generalized:chernoff}) we conclude that $\sum |d^{1/5} \left(X_i - {\mu}_i\right)|$ exceeds $t \|\vectorize{\bar{\sigma}}\|_1$ with probability at most
	$$\exp\left(-(t \|\vectorize{\bar{\sigma}}\|_1)^2/\left(2d^{2/5}d\right)\right) = \exp\left(-\Omega(t^2)\right).$$

	Next, consider Property~\eqref{assumption:scaled_concentrated_generalized}. Define $Y_i = (d^{1/5} (X_i - \probability_i))^2/{\sigma}_i^{4/3}$.
	Note that $Y_i$ takes values in an interval of length $d_i \leq d^{2/5}$. $|\sum Y_i - \mathbf{E}[\sum Y_i]|$ exceeds $t \|\vectorize{\bar{\sigma}}^{2/3}\|_1$ with probability at most
	\begin{align*}
	\exp\left(-\frac{t^2 \|\vectorize{\bar{\sigma}}^{2/3}\|^2_1}{2 \sum_{i=1}^d d_i^2}\right) &\leq \exp\left(-\frac{t^2 \|\vectorize{\bar{\sigma}}^{2/3}\|^2_1}{2 d^{9/5}}\right) \leq \exp\left(-t^2 d^{1/5}/2\right) \\ &= \exp\left(-\Omega(t^2)\right)\enspace ,
	\end{align*}
	which finishes the proof.
\end{proof}

%% file: tex/low-variance.tex
\subsection{Approaches for small variance}\label{sec:low-variance}

Throughout our analysis in Section~\ref{sec:analysis} we assume that $\sigma_i \geq 1$ for all $i$. 
If it is not known whether this assumption holds, one can enforce this by adding independent noise of variance 1 to each coordinate.
If the additional error due to this noise is considered too large, we can proceed in the following ways.
For simplicity, we assume that we know $\vectorize{\sigma}$ exactly, but similar approaches can be used if we only have an approximation $\vectorize{\hat{\sigma}}$.

Let $\min{\{\vectorize{\sigma}\}}$ be the smallest entry in $\vectorize{\sigma}$. 
We assume in this paragraph that all dimensions are non-constant, i.e., $\min{\{\vectorize{\sigma}\}} > 0$.
A simple technique is to scale all inputs and standard deviations by $1/\min{\{\vectorize{\sigma}\}}$ to spread out the data.
That is, $\bar{\sigma_i} = \sigma_i/\min{\{\vectorize{\sigma}\}}$ and $\bar{x}^{(j)}_i = x^{(j)}_i/\min{\{\vectorize{\sigma}\}}$ for all $i$ and $j$.
Run \algorithmname on input $\vectorize{\bar{x}}^{(1)},\dots,\vectorize{\bar{x}}^{(n)}$, and $\vectorize{\bar{\sigma}}^2$, and scale back the private estimate by $\min{\{\vectorize{\sigma}\}}$.

The approach above works well when $\min{\{\vectorize{\sigma}\}}$ is not too small, but is not ideal if some coordinates have tiny variance since it requires a huge scaling.
Instead we can change the parameters for \privatequantile.
We use the assumption that $\min{\{\vectorize{\sigma}\}} \geq 1$ only in Part 1 of Section~\ref{sec:analysis-bounding-mu}.
There we show that the estimate returned by \privatequantile~ is close to a point in $[{\mu}_i - 2{\sigma}_i, {\mu}_i + 2{\sigma}_i]$ with high probability. 
Specifically, the estimate is within distance $1$ of such a point and as such it is in $[{\mu}_i - 3{\sigma}_i, {\mu}_i + 3{\sigma}_i]$ for any ${\sigma}_i \geq 1$.
When ${\sigma}_i < 1$ we can increase the number of steps of the binary search to get closer to the interval.
For example, for a small value $\tau > 0$, we can use $T = \log_2(M/\tau)$ to return an estimate that is within distance at most $\tau$ with high probability.
In the case that there are coordinates ${\sigma}_i < \tau$ so tiny that we do not return a point in $[{\mu}_i - 3{\sigma}_i, {\mu}_i + 3{\sigma}_i]$, this might increase the clipping error and the expected error due to noise. 
For this setting, the analysis in Section~\ref{sec:clipping:error} and Section~\ref{sec:noise:error} can only use the bound $|\vectorize{\mu}_i - \tilde{\vectorize{\mu}}_i| \leq 2\sigma_i + \tau$. This results in an additional clipping error of $\frac{1}{n}\left(k + \sqrt{1/\rho}\right)\tau\sqrt{d}$. Using $\tau \leq \min(\sqrt{n/d}, n\sqrt{\rho/d})$ bounds the additional contribution of this clipping error to be $\tilde{O}(1)$.
To bound the additional error due to noise, note that $\tau \| \hat{\Sigma}^{-1/4}\|_2 = \tau /\sqrt{\|\vectorize{\sigma}\|_1}$, which means that the value of $C$ in Section~\ref{sec:noise:error} is bounded by $\tilde{O}\left(\sqrt{\|\vectorize{\sigma}\|_1} + \tau/\sqrt{\|\vectorize{\sigma}\|_1}\right)$. Carrying out the computation of the expected error due to noise with this value of $C$ yields the bound $\tilde{O}\left((\|\vectorize{\sigma}\|_1^2 + \tau \|\vectorize{\sigma}\|_1 + \tau^2)/\rho_3\right)$.
Setting $\tau \leq \|\vectorize{\sigma}\|_1$, the additional noise error has asymptotically not influence on the expected error bounds in Theorems~\ref{thm:main-detailed} and~\ref{thm:main-detailed-generalized}.

%% file: tex/generic.tex
\section{Generic Bounds in the Absence of Variance estimates}\label{sec:generic-bounds}

\Cref{alg:our-algorithm} requires as input estimates $\vectorize{\hat{\sigma}}^2$ on the coordinate-wise variances.
If the input is $\vectorize{\sigma}$-well concentrated, \Cref{thm:main-detailed}
provided bounds on the expected $\ell_2$ error of the algorithm.
In this section, we consider the case that no such estimates are known and we run the algorithm without the scaling step.
This is similar in spirit to using the ``shifted-clipped-mean estimator''  of Huang et al.~\cite{huang_instance-optimal_2021}.
The following theorem shows that even without carrying out the random rotation (see Section 3.3 in~\cite{huang_instance-optimal_2021}), we match their bounds up to constant terms in expectation.
This result makes their algorithm useful in settings where a random rotation impacts performance negatively, e.g., when vectors are sparse.

Let $D \subseteq \mathbf{R}^d$ be the collection of $n$ vectors $\vectorize{x}^{(1)}, \ldots, \vectorize{x}^{(n)}$.
Let $w(D) = \max_{\vectorize{x}, \vectorize{y} \in D} \| \vectorize{x} - \vectorize{y} \|_2$ be the diameter of the dataset.
For simplicity we assume that $w(D) \geq 1$ and $M = O(w(D))$.
In comparison to the distributional setting studied before, there is no sampling error involved in mean estimation
and we only measure clipping error and the error due to noise.

\begin{theorem}
    \label{thm:diameter:bound}
    Let $d \geq 1, \rho > 0$, and $D$ be of size $n = \Tilde{\Omega}\left(\sqrt{\tfrac{d}{\rho}}\right)$. If Algorithm~\ref{alg:our-algorithm} is run with $k = \tilde\Theta\left(\sqrt{\frac{d}{\rho}}\right)$, $\rho_1=\rho_2=\rho_3=\rho/3$, and ${\hat\sigma}^2_i = 1$ for all $i \in \{1, \ldots, d\}$ then the expected $\ell_2$ error due to clipping and noise is
    \begin{align*}
        \tilde{O}\left(\sqrt{\frac{d}{\rho}} \frac{w(D)}{n}\right) \enspace .
    \end{align*}
\end{theorem}

\noindent Before proceeding to the proof, we introduce some helpful notation.

\begin{definition}
    Let $\vectorize{\mu} \in \mathbf{R}^d$ be the coordinate-wise median of $D$.
    For $0 \leq \alpha < 1/2$, we say that $\vectorize{\tilde\mu} \in \mathbf{R}^d$ is $\alpha$-good if each ${\tilde{\mu}}_i$ has rank error at most $\alpha n$ from ${\mu}_i$.
\end{definition}

\begin{lemma}
 Given $0 \leq \alpha < 1/2$, let $\vectorize{\tilde{\mu}}$ be $\alpha$-good.
 For each $\vectorize{x}^{(j)} \in D$, $1 \leq j \leq n$, $\| \vectorize{x}^{(j)} - \vectorize{\tilde\mu} \|_2^2 \leq \frac{w(D)^2}{1/2 - \alpha}$.
\end{lemma}

\begin{proof}
    Fix $\vectorize{x}^{(j)}$ and compute
    \begin{align*}
        (n - 1) w(D)^2 &\geq \sum_{j \neq j'} \| \vectorize{x}^{(j)} - \vectorize{x}^{(j')} \|_2^2\\
        &= \sum_{i = 1}^d \sum_{j \neq j'} \left({x}_i^{(j)} - {x}_i^{(j')}\right)^2\\
        &\geq \sum_{i = 1}^d (1/2-\alpha) (n - 1) \left({x}_i^{(j)} - {\tilde\mu}_i\right)^2\\
        &=  (1/2 - \alpha)(n - 1) \|\vectorize{x}^{(j)} - \vectorize{\tilde\mu}\|_2^2.
    \end{align*}
    The lemma follows by re-ordering terms.
\end{proof}

\begin{proof}[Proof of \Cref{thm:diameter:bound}]
    We instantiate \privatequantile as the binary search based method of~\cite{huang_instance-optimal_2021} with $T=\log(M n \sqrt{\rho})$.
    Fix some value for $\alpha$, say 1/3.

    By Lemma~\ref{lemma:approximate-quantile} the coordinate-wise median computed on \Cref{alg:center-prediction} of \Cref{alg:our-algorithm} is within $\ell_2$ distance $\sqrt{d/\rho}/n$ of an $\alpha$-good point with high probability.
    Let us assume that $\vectorize{\tilde{\mu}}$ is itself $\alpha$-good as the additional error of $\sqrt{d/\rho}/n$ from the quantile search is dominated by the error from clipping and noise.
    The failure probability from Lemma~\ref{lemma:approximate-quantile} must be set sufficiently low such that the expected error from failure events is also dominated by the error from clipping and noise.
    Since the maximal possible $\ell_2$ error is $\sqrt{d}2M$ this is trivially true when $n$ is large enough such that $\beta \leq \frac{w(D)}{\sqrt{\rho}Mn}$.

    Since $\vectorize{\tilde\mu}$ is $\alpha$-good, the maximum length of a shifted vector $\vectorize{x}^{(j)} - \vectorize{\tilde\mu}$ is $O(w(D))$, so $C = O(w(D))$ on \Cref{alg:estimate:clipping} of \Cref{alg:our-algorithm}.
    As in \Cref{sec:noise:error}, the expected $\ell_2$ error due to noise is at most $1/n \cdot \sqrt{2C^2 ( \sum_{i = 1}^d {\hat{\sigma}}_i) / \rho} = \tilde{O}\left(\sqrt{\frac{d}{\rho}} \frac{w(D)}{n}\right)$.

    Analogously to the calculations carried out in \Cref{sec:clipping:error}, we can bound the error due to clipping for all vectors $\|\vectorize{y}^{(j)}\|_2 > C$.
    Setting $k = \tilde\Theta(\sqrt{d/\rho})$ shows that this clipping error is $\tilde{O}\left(\sqrt{d/\rho} \cdot w(D)/n\right).$
\end{proof}

%% file: tex/experiments.tex
\section{Empirical Evaluation}\label{sec:experiments}

To put \algorithmname's utility into context, we measure error in diverse experimental settings.
We use the empirical mean, i.e., the sampling error, as a baseline, since it reflects an inevitable lower bound.
Additionally we compare \algorithmname to \instanceoptimallong~\citep{huang_instance-optimal_2021}, which has been shown to perform at least as well as \coinpress~\cite{biswas_coinpress_2020} in empirical settings~\cite{huang_instance-optimal_2021}; hence representing the current state-of-the-art for differentially private mean estimation.
Our experimental evaluation focuses on the examples of $(\vectorize{\sigma}, p)$-well concentrated distributions considered in Section~\ref{sec:well-concentrated};
we evaluate the utility-privacy tradeoff for Gaussian and binary data under $\ell_2$ and $\ell_1$ error, respectively.
We run our experiments with synthetic data as input.
For the binary case, we also evaluate our error on the Kosarak dataset~\cite{benson_discrete_2018} which represents user visits (or, conversely, non-visits) to webpages, as well as the Point of Sale (POS) dataset\footnote{\url{https://github.com/cpearce/HARM/blob/master/datasets/BMS-POS.csv}} which represents user purchases.
The code used for our experiments is available at \href{https://github.com/ChristianLebeda/PLAN-experiments}{https://github.com/ChristianLebeda/PLAN-experiments}.

\paragraph{Additional experiments} Due to space constraints, we provide additional experimental results in the appendices. 
\Cref{app:correlations} evaluates the performance of \algorithmname and \instanceoptimalshort for correlated, Gaussian data. 
We demonstrate that \algorithmname is less affected by dependencies than \instanceoptimalshort.
\Cref{app:no-scaling} includes the variant of \algorithmname that does not scale the input, cf.~Section~\ref{sec:generic-bounds}, and compares it to \algorithmname and \instanceoptimalshort. In a nutshell, this variant recovers the behavior of \instanceoptimalshort even in empirical settings in which constant matters, as suggested by Theorem~\ref{thm:diameter:bound}. 
Finally, Appendix~\ref{app:gaussian-2048} discusses experiments for larger input sizes, which technically is far away from our assumptions on the input size.
The evaluation shows that \algorithmname is less robust than its competitors for largest inputs and smallest privacy budget.

\subsection{Implementation}\label{sec:implementation}
We evaluate the empirical accuracy of \algorithmname\ by instantiating \Cref{alg:our-algorithm} in Python 3.
Our implementation contains two different instantiations of \algorithmname: one version for binary data targeting $p=1$, and one version for data from multivariate Gaussian distributions for $p=2$.
The pseudo code for both instantiations of \algorithmname is shown in \Cref{code:plan}.
Both instances use the \privatequantile search by \citet{kaplan_differentially_2022}.
The implementation of \instanceoptimalshort uses the original source code from \citet{huang_instance-optimal_2021}.

\begin{center}
\begin{minipage}[htb]{\columnwidth}
\begin{lstlisting}[language=plan,label=code:plan,caption={Pseudo code for the \algorithmname instantiation for $n$ vectors in $\mathbf{R}^d$, with each coordinate being in the
	range ${[-M, M]}$, targeting the expected $\ell_p$ error with privacy budget $\rho$ and a failure probability of at most $\beta$.
	Note that budget is spent on estimating the standard deviation unlike in \Cref{alg:our-algorithm} where \vectorize{\hat{\sigma}^2} is an input parameter.
	}]
def PLAN(data, n, d, M, p, rho, beta) {
 rho1, rho2, rho3 = divideBudget(rho)
 mu = center(M, rho1*0.25, data, beta/3)
 std = estimateStd(M, rho1*0.75, data, beta/3)
 scaleFactors = std**(-2/(p + 2))
 y = (data-mu) * scaleFactors #coordinate-wise
 k = sqrt(n) + rankError(M, n, d, rho2, beta/3)
 z = clip(M, y, rho2, (n-k)/n)
 return ((z+noise(p, rho3))/scaleFactors)+mu
}
\end{lstlisting}
\end{minipage}
\end{center}

\paragraph{Estimating $\vectorize{\sigma}^2$.}
We recall that \Cref{alg:our-algorithm} assumes an estimate of the variances as input.
In the absence of public knowledge, these parameters have to be estimated on the actual data in a differentially private way, as mentioned in \Cref{code:plan}. We find estimates in the following way.

Let $X, Y \sim \mathcal{D}$ be two i.i.d. random variables. By a standard calculation, 
$\mathbf{E}[(X - Y)^2/2] = \mathbf{E}[X^2] - \mathbf{E}[X]^2 = \mathbf{Var}[X]$.
In the Gaussian case, given $X, Y \sim \mathcal{N}({\mu_i}, \sigma_i^2)$, $(X-Y)^2/2$ follows a Chi-squared distribution with one degree of freedom. 
We use \privatequantile for each coordinate to differentially privately estimate the median $\tilde{\mu}_i$ of a sequence of values $(X-Y)^2/2$.
Using the Wilson-Hilferty transformation~\cite[(18.24)]{johnson1994chi}, we then estimate the mean as $\tilde{\mu_i}/(9/7)^3$ as a post-processing step.
In the binary case, each coordinate is 1 with probability $\probability_i$. 
We estimate the variance by privately estimating the mean $\tilde{\probability}_i$ using the Gaussian mechanism with $\ell_2$-sensitivity $1/n$, and use ${\hat{\sigma}}^2_i = \tilde{\probability}_i (1 - \tilde{\probability}_i)$.
In both cases, we regularize the estimate $\vectorize{\hat{\sigma}}$ on the standard deviation by adding $\|\vectorize{\hat{\sigma}}\|_1 / d$ to each coordinate.

In \Cref{app:std-evaluation}, we generalize the estimator on Gaussian data to distributions with bounded fourth moment and provide
more details on the empirical evaluation of differentially private variance estimation.

\paragraph{Bounding the clipping universe.}
Scaling the data based on the variances estimate $\hat\sigma^2$ gives an opportunity to trim the universe used when searching for the clipping radius (\Cref{alg:our-algorithm}, \Cref{alg:estimate:clipping}).
For the $\ell_2$ error, instead of using $M\sqrt{d}$ as the upper bound to cover the entire universe, we use the bound $\sqrt{\log(n)\log(1/\beta)\|\vectorize{\hat\sigma}\|_1}$ for a given failure probability $\beta$.
We set this failure probability to 0.1.
We can use this bound because Assumption~\eqref{assumption:scaled_concentrated} in \Cref{def:well-concentrated} guarantees that the length of vectors after scaling is proportional to $\sqrt{\|\vectorize{\hat\sigma}\|_1}$.

\subsection{Experiment design}\label{sec:experiment-design}

\begin{table*}[tb]
	\begin{tabularx}{\linewidth}{lllllll}
		\toprule
		\textbf{Name} &  $n$ & $\rho$ & \universevariable  &  $d$ & \vectorize{\sigma^2} & \vectorize{\mu}\\
		\midrule
		\casethem & $4000$ &  $\{1, 0.5, 0.125\}$  & $50\sqrt{d}$ &  $\{16, 32, \compactdots, 1024\}$ & $\left(1, \compactdots, 1\right)$ & $0^d$ \\
		\caseplausible & $10\,000$ & $\{1, 0.125\}$ & \pessimisticuniverse &  $512$ & $\left(\left(\frac{d}{d}\right)^\alpha, \left(\frac{d}{d-1}\right)^\alpha, \compactdots , \left(\frac{d}{1}\right)^\alpha\right)$ & $10^d$ \\
		\caseshine & $10\,000$ & $\{1, 0.5, 0.125\}$& \pessimisticuniverse &  $\{16, 32, \compactdots, 1024\}$  & $\left(\left(\frac{d}{d}\right)^2, \left(\frac{d}{d-1}\right)^2, \compactdots , \left(\frac{d}{1}\right)^2\right)$ & $10^d$ \\
		\midrule
		\casebinary &  $4096$ & $\{1, 0.5, 0.125\}$ & $\|\vectorize{\hat{\sigma}^{2/3}}\|_2$ &  $\{256, 512, \compactdots, 2048\}$ & $\sigma_i^2 = \mu_i(1-\mu_i)$ & $\mu_i = \begin{cases} \tfrac{1}{2} & i \leq \alpha d \\ \tfrac{1}{100} & \text{else} \end{cases}$\\
		\midrule
		\casekosarak &  $75\,462$ & $\{1, 0.5, 0.25, 0.125, 0.0625\}$ & $\|\vectorize{\hat{\sigma}^{2/3}}\|_2$ &  $27\,983$ & \multicolumn{2}{c}{\emph{Fixed dataset}}  \\
		\casepos &  $515\,596$ & $\{1, 0.5, 0.25, 0.125, 0.0625\}$ & $\|\vectorize{\hat{\sigma}^{2/3}}\|_2$ &  $1\,657$ & \multicolumn{2}{c}{\emph{Fixed dataset}} \\
		\toprule
		\multicolumn{7}{c}{Global setting: $\rho_1=0.25\rho/d, \rho_2 = 0.25(\rho-\rho_1 ), \rho_3=\rho-\rho_1-\rho_2$.} \\
		\bottomrule
	\end{tabularx}
	\caption{Settings for the experiments. All experiment settings run 50 times for each algorithm. \instanceoptimalshort is run with 20 steps for the binary quantile selection. Note that $\rho_1$ is split between recentering and variance prediction for \algorithmname.}
	\label{tab:experiment-settings}
\end{table*}

\paragraph{Parameter input space.}
The following parameters need to be chosen for each execution of \algorithmname: the universe $\universevariable$, the $\ell_p$ error norm, and the privacy budget $\rho$ as well as the partitioning of $\rho$ into $\rho_1, \rho_2, \rho_3$.
We also use a failure probability $\beta$ to compute exact constants for the universe after scaling and the value of $k$ for the clipping threshold quantile, cf. \Cref{algo:set:parameters} in \Cref{alg:our-algorithm} and the discussion above.
Since \instanceoptimalshort uses a binary search for their quantile selection, the amount of steps to use also needs to be chosen.
\instanceoptimalshort sets the amount of steps to 10 by default, but an empirical investigation shows that this value is too low for many of our settings --- the binary quantile search ends early which causes inaccurate results.
To level the playing field, we use 20 steps to ensure that \instanceoptimalshort does not suffer any disadvantages from the binary quantile search failing. 
Furthermore, we noticed that in the implementation of~\instanceoptimalshort, the quantile for the clipping threshold is set too aggressively. We used as offset a value that is twice as large as their proposal, which made their implementation much more robust without increasing the median error.

\paragraph{Input data.}
We use both synthetic and real-world datasets to evaluate \algorithmname.
When generating synthetic data, the following parameters need to be chosen: the dataset size $n$, the dimensionality $d$, the means \vectorize{\mu}, and the variances \vectorize{\sigma^2}.
Since \algorithmname and \instanceoptimalshort both ignore potential correlations in data, we use covariance matrices of the form $\Sigma=\text{diag}(\vectorize{\sigma}^2)$ for the Gaussian case.
As shown in Appendix~\ref{app:correlations}, the trends observed for the diagonal case carry over to the non-diagonal, correlated case. 
In the binary case for real-world data, coordinates are correlated.

\subsubsection{Gaussian data}\label{sec:design-gaussian-data}
To show the effectiveness of \algorithmname on Gaussian data, we design three diverse experiments.
The first experiment (\casethem) reflects the parameter settings used in previous work by \citet{huang_instance-optimal_2021} where data has no skew, which is the worst case for our proposed method.
The second experiment (\caseplausible) simulates data ranging from no to significant skew across dimensions, showcasing how \algorithmname improves with increasing skew.
Finally, the third experiment (\caseshine) highlights how \algorithmname scales as dimensionality increases for data with a skew.
For each experiment, we vary $\rho$ between $1$ and $0.01$ to show how accuracy scales in higher and lower privacy regimes.
We summarize the settings used in the experiments in \Cref{tab:experiment-settings}.
Note that the settings $\rho=0.01$ corresponds to ($\varepsilon\approx 0.75$, $\delta\approx 10^{-6}$)-approximate differential privacy.

\paragraph{Budget division.}

The values we report for $\rho$ in \Cref{tab:experiment-settings} correspond to the total privacy budget for the entire experiment, also including variance estimation.
Our algorithm needs to perform two preprocessing steps: estimating  \vectorize{\mu} for re-centering, and estimating  \vectorize{\sigma^2} for scaling the noise.
We fix the initial estimation of \vectorize{\mu} and  \vectorize{\sigma^2} to use 25\% of the total privacy budget --- the same proportion used for preprocessing as in \citet{huang_instance-optimal_2021}.
In the same spirit, we set the budget to determine the clipping threshold ($\rho_2$) to 25\% of the remaining budget, and use the larger part ($\rho_3$) for the Gaussian noise.

\paragraph{Choosing valid settings.}
Just like for \instanceoptimalshort, \universevariable\ needs to be set such that \vectorize{\mu} is within the universe.
We will use two different approaches to set \universevariable: the approach from \cite{huang_instance-optimal_2021} ($\universevariable=50\sqrt{d}$), and a more pessimistic approach where we assume all values have the worst-case standard deviation across all dimensions, and create more leeway by scaling with a constant ($\universevariable=\pessimisticuniverse$).

Additionally, the rank error needs to be tuned such that \privatequantile search can be expected to return a quantile close to the requested one.
Since \algorithmname calls \privatequantile multiple times, \algorithmname needs to tolerate the worst-case rank error for all calls.
We set $n$ such that the rank error is at most $0.1n$ for each value of $\rho$.

\ourparagraph{\casethem: }{no skew}
To show that \algorithmname\ performs comparatively to \instanceoptimalshort we run it on data where variance is the same across all dimensions.
In this setting we expect \algorithmname to perform similarly to \instanceoptimalshort.
We reuse the experiment settings used by \citet{huang_instance-optimal_2021} for a fair comparison.

\ourparagraph{\caseplausible: }{varying skewness}
To show how \algorithmname improves as the input data's skew increases, we vary the skewness of the variance.
We introduce a parameter $\alpha$, and simulate a Zipfian like skew to the data, and set the variances $\vectorize{\sigma^2}=((d/d)^\alpha, (d/(d-1))^\alpha, \compactdots , (d/1)^\alpha)$ for $\alpha\in\{0, 0.5, \compactdots, 2\}$.
In this setting we expect \algorithmname to outperform \instanceoptimalshort for $\alpha> 0$.

\ourparagraph{\caseshine: }{varying dimensionality}
To show how \algorithmname's advantage scales compared to \instanceoptimalshort, we vary dimensionality as we expect an improvement up to a factor $\sqrt{d}$.
Since \algorithmname's advantage is based on data having a skew, we set $\vectorize{\sigma^2}=\left((d/d)^2, (d/(d-1))^2, \compactdots , (d/1)^2\right)$.
Note that \algorithmname's improvement is in \emph{noise error}, as sampling error is unavoidable --- we compute the error relative to the empirical mean in this experiment to showcase the difference in noise error.

\subsubsection{Binary data}\label{sec:design-binary-data}

\begin{figure*}[ht]
	\includegraphics*[width=0.9\textwidth]{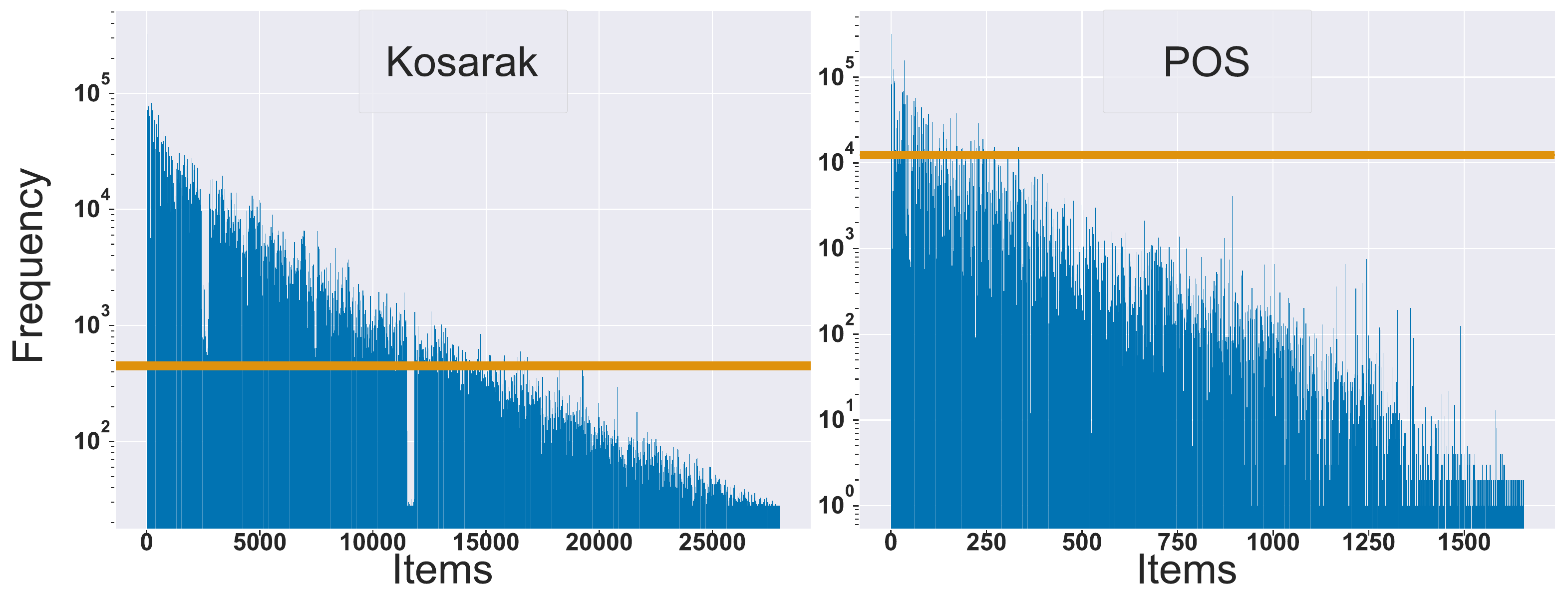}
	\caption{Histogram for \casekosarak (left) and \casepos (right).
	The orange line is the smallest allowed variance according to \Cref{lem:binary:concentrated} which we clip to.}
	\label{fig:histogram:binary}
\end{figure*}

\chtodo{What does the line represent? What does it mean that we clip to the variance?}

To diversify our experimental scope, we consider binary data represented by $n$ bitvectors of length $d$ in which each bit $i, 1 \leq i \leq d,$ is set independently to 1 with probability $\probability_i$.
We usually think about these bitvectors as sets representing a selection of items from $\{1, \ldots, d\}$.
To vary the error measure, we focus on the $\ell_1$ error. This is akin to computing the total variation distance (TVD) $1/2 \| x - y \|_1$, but we avoid the normalization of $x$ and $y$ to have unit $\ell_1$ norm.

We design three experiments: the first (\casebinary) varies the skewness in the probabilities to make controlled experiments on the accuracy of \algorithmname.
The two remaining experiments (\casekosarak, \casepos) use real-world datasets that naturally exhibit skew between coordinates.

\ourparagraph{\casebinary: }{Varying skewness}
This experiment follows the same design principle as \caseplausible: to show how skewness affects the performance of \algorithmname.
Given $n, d,$ and $\rho$, we choose two probabilities $\probability_1 = 0.5$ (high variance) and $\probability_2 = 0.01$ (low variance).
Given $\alpha \in [0, 1]$, we sample the first $\lceil \alpha d\rceil$ bits with probability $\probability_1$ each, and the remaining positions with probability $\probability_2$.
The low variance setting is slightly below the minimum threshold of $1/d^{2/5}$ discussed in~\Cref{lem:binary:concentrated} to test the robustness of our implementation. We clip all estimated variances to $1/d^{2/5}$ from below.
\chtodo{I used indexing subscripts here. Could also be $\probability_\alpha$ and $\probability_\beta$.}

\ourparagraph{\casekosarak: } {Website visits}\label{sec:design-kosarak}
The \emph{Kosarak dataset}\footnote{\url{http://fimi.uantwerpen.be/data/}} represents click-stream data of a Hungarian news portal.
There are $n = 75\,462$ users and a collection of $d = 27\,983$ websites.
In total, users clicked on 4\,194\,414 websites (each user clicked on 55.6 websites on average), and there is a large skew between the websites, see \Cref{fig:histogram:binary}.

\ourparagraph{\casepos: }{Shopping baskets}\label{sec:design-pos}
The \emph{POS dataset} %
contains merchant transactions on $d = 1\,657$ categories from $n = 515\,596$ users.
In total, there are $3\,367\,019$ transactions (around 6.5 on average per user).
Again, there is large skew in the different categories, see \Cref{fig:histogram:binary}.
The dataset is particularly challenging because the minimum variance $d^{-2/5}$ (cf. \Cref{lem:binary:concentrated}) has to be clipped on many coordinates.

\begin{figure*}[tb]
	\centering
	\includegraphics[width=\linewidth]{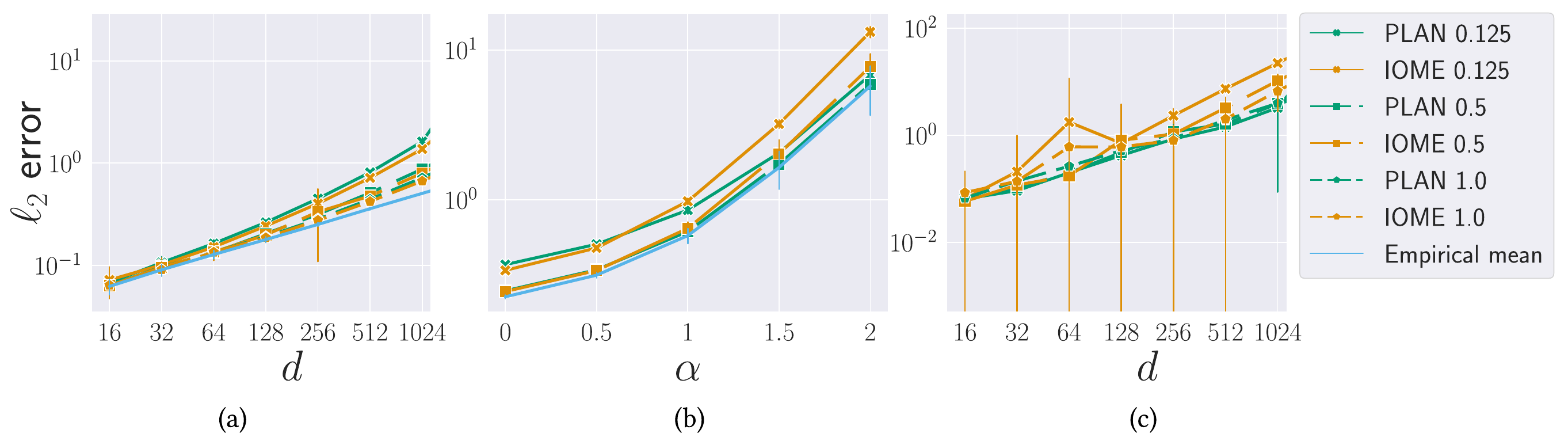}
	\caption{$\ell_2$ error for synthetic Gaussian data when varying (a) dimensions with data without a skew,
		(b) skewness of the variances, and
		(c) dimensions  for skewed data --- note that we compute error relative to the empirical mean rather than the statistical mean in this experiment as sampling error dominates in this setting. Also notice the different scales on the y-axis.}
	\label{fig:gaussian-experiments}
\end{figure*}

\begin{figure*}[tb]
	\centering
	\includegraphics[width=\linewidth]{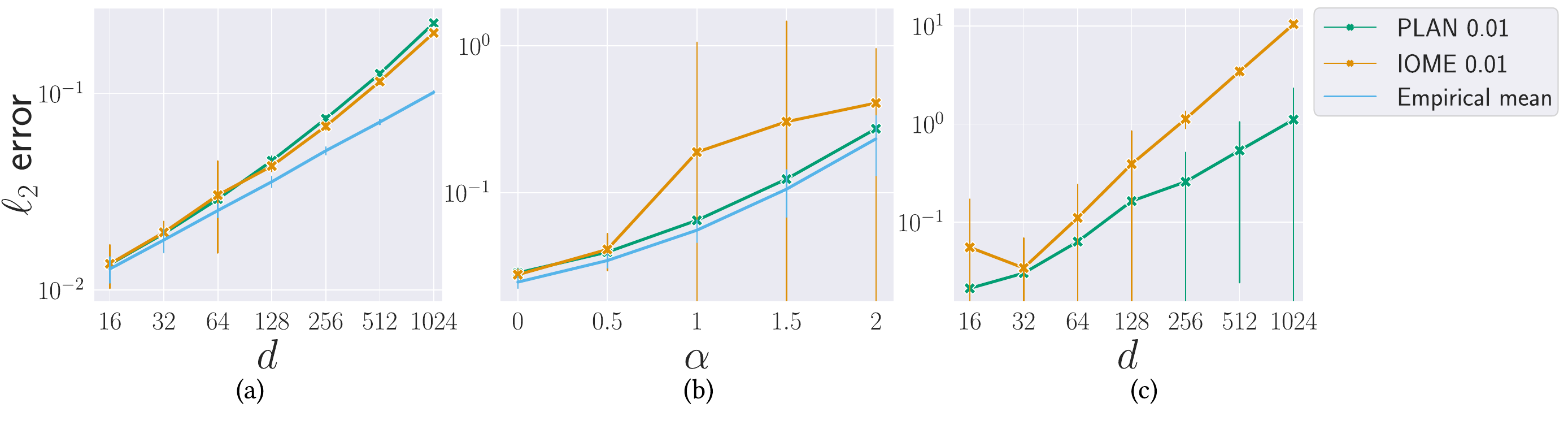}
	\caption{$\ell_2$ error for \casethem, \caseplausible, and \caseshine with $\rho=0.01$ which implies $(\varepsilon, \delta)$-approximate DP with $\varepsilon < 1$ for $\delta \approx 10^{-6}$.
		To tolerate the rank error, we reuse the parameters from the original experiments and set $n= 100\, 000$, and $d=64$ for (b).}
	\label{fig:gaussian-experiments-tiny-rho}
\end{figure*}

\begin{figure*}[tb]
	\centering
	\includegraphics[width=\linewidth]{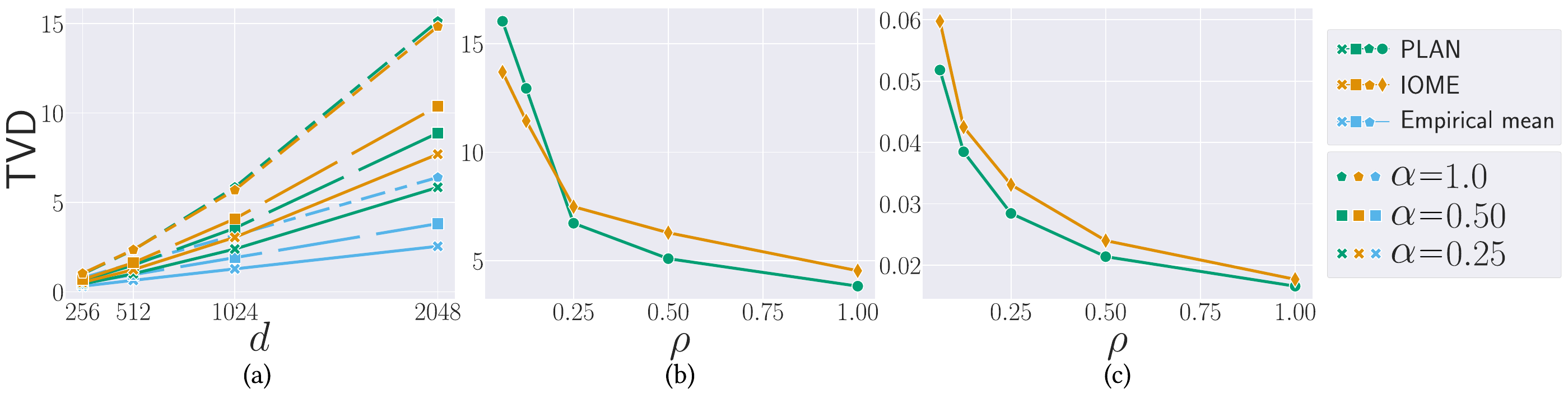}
	\caption{(a) Synthetic binary data, varying the ratio of 0s to 1s
		(b) Kosarak dataset
		(c) POS dataset}
	\label{fig:binary-experiments}
\end{figure*}

\subsection{Results}\label{sec:experiment-results}

For the Gaussian case, we ran our experiments using Python 3.11.3 on a MacBook Pro with 24GB RAM, and the Apple M2 chip (8-core CPU).
For the binary case, we had to run the experiments on a more powerful machine to support \instanceoptimalshort on the real-world datasets.
While Kosarak and POS are sparse datasets, \instanceoptimalshort requires that the entire dataset (not just the sparse representation) is loaded into memory to perform a random rotation the algorithm uses as a preprocessing step.
As a consequence, we ran the binary experiments using Python 3.6.9 on a machine with 512GB RAM, on 2x Intel(R) Xeon(R) CPU E5-2690 v4 @ 2.60GHz (2x 14-core CPU).
Even on this more powerful machine with multiprocessing using all cores, a single run of \instanceoptimalshort took at least 26 minutes on Kosarak, and at least 9 minutes on POS.
In comparison, \algorithmname spent an average of 12 and 9 seconds running on Kosarak and POS, respectively.

\subsubsection{Gaussian data}\label{sec:results-gaussian-data}

All results are shown in \Cref{fig:gaussian-experiments} and \Cref{fig:gaussian-experiments-tiny-rho}.
Additionally, in \Cref{app:gaussian-2048} we also include $d=2048$ for the experiments.
\Cref{fig:gaussian-experiments}~(a) shows \casethem, where \algorithmname and \instanceoptimalshort have comparable accuracy.
The same behavior is observed for smaller $\rho$, as shown in \Cref{fig:gaussian-experiments-tiny-rho}~(a).
Both plots show expected behavior that aligns with our theoretical results.

\Cref{fig:gaussian-experiments}~(b) shows \caseplausible, where \algorithmname performs better than \instanceoptimalshort for $\alpha>1$.
For $\alpha\leq1$ the error between \algorithmname and \instanceoptimalshort is similar.
This is expected behavior, as $\alpha=0$ represents the same input data as in \casethem.
Notice how \algorithmname approaches the empirical mean as $\alpha$ grows for both $\rho=0.5$ and $\rho=1$.
For smaller values of $\rho$, \instanceoptimalshort is less robust, as can be seen by the error bars in \Cref{fig:gaussian-experiments-tiny-rho}~(b), whereas \algorithmname retains robust error for all values of $\alpha$ using the same settings.

\Cref{fig:gaussian-experiments}~(c) shows \caseshine, where we compare against the empirical mean since sampling error is larger than the noise error for \algorithmname\ in this case.
As expected, \algorithmname increases its advantage over \instanceoptimalshort as $d$ grows.
Since we have subtracted the empirical mean the error scale is smaller, which is why we observe larger error bars in this setting.
The difference between \instanceoptimalshort and \algorithmname's error is further showcased in \Cref{fig:gaussian-experiments-tiny-rho}~(c), where \algorithmname has a lower error, but the variance increases with $d$.

\subsubsection{Binary data}\label{sec:results-binary-data}

All results are shown in \Cref{fig:binary-experiments}.
\Cref{fig:binary-experiments}~(a) shows \casebinary, where \algorithmname has an advantage over \instanceoptimalshort for $\alpha<1$ which increases as $\alpha$ decreases.
This is the expected behavior, as \algorithmname is able to exploit the skew in variance whereas \instanceoptimalshort treats every dimension the same.

\Cref{fig:binary-experiments}~(b) shows \casekosarak.
As we can see, \algorithmname outperforms \instanceoptimalshort for sufficiently large values of $\rho$.
For small $\rho$ ($\rho\leq0.125$), \algorithmname is running in an invalid setting --- our assumptions on rank error are not fulfilled in these cases.

\Cref{fig:binary-experiments}~(c) shows \casepos.
\algorithmname has a slight advantage compared to \instanceoptimal in this case, which decreases as $\rho$ grows.

%% file: tex/related-work.tex
\section{Related Work}\label{sec:related-work}
Our work builds on concepts from multiple areas within the literature on differential privacy.
We provide an overview of the most closely related work.

\paragraph{Statistical private mean estimation.}
There is a large, recent literature on statistical estimation for $d$-dimensional distributions under differential privacy, mainly focusing on the case of Gaussian or subgaussian distributions~\citep{alabi_privately_2022, ashtiani_private_2022, biswas_coinpress_2020, brown_covariance-aware_2021, brown_fast_2023, du_differentially_2020, duchi_fast_2023-1, hopkins_efficient_2022, karwa_finite_2018, kamath_privately_2019, kamath_private_2022, kothari_private_2022}.
The error on mean estimates is generally expressed in terms of Mahalanobis distance, which is natural if we want the error to be preserved under affine transformations.
Some of these efficient estimators are even \emph{robust} against adversarial changes to the input data~\citep{alabi_privately_2022, kothari_private_2022}.
Other estimators work even for rather heavy-tailed distributions~\citep{kamath_private_2020}.
What all these estimators have in common is that the nominal dimension $d$ influences the privacy-utility trade-off such that higher-dimensional vectors have a worse trade-off.
To our best knowledge, the algorithm among these that has been shown to work best in practical (non-adversarial) settings is the \coinpress algorithm of \citet{biswas_coinpress_2020}.

\paragraph{Adapting to the data.}
The best private mean estimation algorithms are near-optimal for worst-case $d$-dimensional distributions in view of known lower bounds~\citep{cai_cost_2021}.
However, it is natural to consider ways of improving the privacy-utility trade-off whenever the input distribution has some structure.
One way of going beyond the worst case is by privately identifying low-dimensional structure~\citep{amin_differentially_2019,dwork_analyze_2014, hardt_noisy_2014, singhal_privately_2021-1}.
Such methods effectively reduce the mean estimation problem to an equivalent problem with a dimension smaller than $d$.
However, we are not aware of any work showing this approach to be practically relevant for mean estimation.

Another approach for adapting to the data is \emph{instance optimality}, introduced by~\citet{asi_instance-optimality_2020} and studied in the context of mean estimation by \citet{huang_instance-optimal_2021} who use $\ell_2$ error (or mean square error) as the utility metric.
The goal is optimality, i.e.~matching lower bounds, for a class of inputs with a given diameter but no further structure.
\citet{huang_instance-optimal_2021} found that their private mean estimation algorithm often has smaller error than \coinpress in practice.
Because of this, and since they also aim to minimize an $\ell_p$ error, this algorithm was chosen as our main point of comparison.

Neither of the mentioned approaches takes \emph{skew} in the data distribution into account, so we believe this is a novel aspect of our work in the context of mean estimation.
However, we mention that privacy budgeting in skewed settings has recently been studied in the context of multi-task learning~\citet{krichene_multi-task_2023}.
Also, the related setting of mean estimation with heterogeneous data (where the sensitivity with respect to different clients' data can differ) was recently studied by~\citet{cummings_mean_2022}.

\paragraph{Clipping.}

An important aspect of private mean estimation for unbounded distributions, in theory and practice, is how to perform \emph{clipping} to reduce the sensitivity.
This has in particular been studied in the context of differentially private stochastic gradient descent~\citep{mcmahan_learning_2018, pichapati_adaclip_2019,andrew_differentially_2021,bu_automatic_2022}.
Though clipping introduces bias, \citet{kamath_bias-variance-privacy_2023} have shown that this is unavoidable without additional assumptions.

\citet{huang_instance-optimal_2021} used a clipping method designed to cut off a carefully chosen, small fraction of the data points.
The clipping done in \algorithmname follows the same pattern, though it is applied only after carefully scaling data according to the coordinate variances.
Thus, it corresponds to clipping to an axis-aligned ellipsoid.

To formally bound clipping error one can either express the error in terms of the diameter of the dataset or analyze the error under some assumption on the data distribution. Both approaches are explored in~\citet{huang_instance-optimal_2021}, but in this paper we have chosen to focus on the latter.

%% file: tex/conclusion.tex
\section{Conclusion and future work}\label{sec:conclusion}

We introduce \algorithmnamelong, a family of algorithms for differentially private mean estimation of $d$-dimensional data.
\algorithmname exploits skew in data's variance to achieve better $\ell_p$ error.
In the case of $\ell_2$ error we achieve a particularly clean bound, namely error proportional to $\|\vectorize{\sigma}\|_1$.
This is never worse than the error of $\sqrt{d}\|\vectorize{\sigma}\|_2$ obtained by previous methods and gives an improvement up to a factor of up to $\sqrt{d}$ when $\vectorize{\sigma}$ is skewed.
While the privacy guarantees hold for any input, the error bounds hold for independently sampled data from distributions that follow a well-defined assumption on concentration.

Finally, we implement two \algorithmname instantiations and empirically evaluate their utility.
Practice follows theory --- \algorithmname outperforms the current state-of-the-art for skewed datasets, and is able to perform competitively for datasets without skewed variance.
To aid practitioners in implementing their own \algorithmname, we summarize some practical advice based on our lessons learned.

\paragraph{Advice for practitioners.}
When implementing a \algorithmname instantiation, practitioners should pose the following questions:
\begin{enumerate}[ref={Question~\arabic*}]
	\item\label{question:estimator} Is there a suitable estimator for the variances ${\sigma}_i^2$ of the data distribution?
	\item\label{question:universe-bound} Can a tighter bound on the clipping universe ($M\sqrt{d}$) be used?
	\item\label{question:robustness} How robust is \algorithmname for my given settings, i.e., will the rank error be too high and cause \algorithmname to fail?
\end{enumerate}
We give examples of how to answer these questions in our evaluation.
To answer \ref{question:estimator}, we derive \emph{private} variance estimators tuned to the data distribution.
As for \ref{question:universe-bound}, when the distribution is $(\vectorize{\sigma}, p)$-well concentrated, much better bounds on the universe size can be derived by using Assumption~\eqref{assumption:scaled_concentrated_generalized} in \Cref{def:well:concentrated:generalized}.
Finally, answering \ref{question:robustness}, robustness must carefully be evaluated using the assumption on minimum $\rho$ values and maximum dimensionality $d$.
Both parameters in conjunction give minimum requirements to the required sample size $n$.

\paragraph{Future work.} We conclude with some possible future directions.
\begin{itemize}
    \item An interesting avenue to explore would be to capture skew in the input vector that is not necessarily visible in the standard basis. For example, one could use private PCA to rotate the space into a basis in which coordinates are nearly independent, and then apply \algorithmname.
    \item Though our algorithm chooses optimal parameters within a class of mechanisms, we have not ruled out that an entirely different approach could have better performance. We conjecture that for any choice of $\vectorize{\sigma}$ there exists an input distribution for which our mechanism achieves an optimal trade-off up to logarithmic factors.
    \item Finally, we rely on having access to reasonable estimates of coordinate variances (the diagonal of the covariance matrix). 
    It would be interesting to study this problem in its own right.
    Covariance estimation is a common problem in differential privacy (see e.g. \cite{kamath_privately_2019}),
    but it is likely that estimating the entire covariance matrix is strictly harder than estimating the diagonal.
\end{itemize}

%% file: tex/supplemental.tex
\appendix
\section{Useful statements from probability theory}

\begin{lemma}[Bernstein's inequality]
    \label{lem:bernstein}
    Let $X_1, \ldots, X_n$ be independent zero-mean random variables.
    Suppose that $|X_i| \leq M$ almost surely for all $i$.
    Then for all $t > 0$,
    \begin{align*}
        \Pr\left[\sum_{i = 1}^n X_i \geq t\right] \leq \exp\left(-\frac{t^2/2}{\sum_{i = 1}^n \mathbf{E}[X_i^2] + Mt/3}\right).
    \end{align*}
\end{lemma}

\begin{lemma}[{\cite[Equation (18)]{winkelbauer2012moments}}]
    \label{lem:expected:error:gaussian}
    The $p$th absolute moment of a zero-centered Gaussian distribution for any $p > 0$ is
    \begin{align*}
        \mathbf{E}\left[\left|\mathcal{N}\left(0, \sigma^2\right)\right|^p\right] = \sigma^p \cdot \frac{2^{p/2}\Gamma\left(\tfrac{p + 1}{2}\right)}{\sqrt{\pi}}.
    \end{align*}

\end{lemma}

\begin{lemma}[Generalized Chernoff-Hoeffding Bound~\cite{dubhashi:panconesi:2009}]
	\label{lemma:generalized:chernoff}
	Let $X := \sum_{1 \leq i \leq n} X_i$ where $X_i, 1 \leq i \leq n$ are independently distributed in $[a_i, b_i]$ for $a_i, b_i \in \mathbf{R}$.
	Then for all $t > 0$
	$$\Pr[|X-\mathbf{E}[X]| \geq t] \leq 2\exp\left(\frac{-t^2/2}{\sum_i \left(a_i - b_i\right)^2}\right).$$
\end{lemma}

\section{Proof of Theorem~\ref{thm:main-detailed-generalized}}\label{app:generalized:lp}

\begin{theorem}
	For sufficiently large $n = \tilde{\Omega}(\max(\sqrt{d/\rho}, \rho^{-1}))$, \algorithmname with parameters $k=\sqrt{n}$ and $\rho_1 = \rho_2 = \rho_3 = \rho/3$ is $\rho$-zCDP.
	If inputs are independently sampled from a $(\vectorize{\sigma},p)$-well concentrated distribution, the mean estimate has expected $\ell_p$ error
	\[
	\tilde{O}\left(1 + \frac{||\vectorize{\sigma}||_p}{\sqrt{n}} + \frac{\|\vectorize{\sigma}\|_{2p/(p+2)}}{n\sqrt{\rho}}\right),
	\]
	 where $\tilde{O}$ suppresses polylogarithmic dependencies on $1/\beta$, $n$, $d$, and the bound $M$ on the $\ell_\infty$ norm of inputs.
\end{theorem}

\begin{proof}
We proceed in the same three steps as in the proof of Theorem~\ref{thm:main-detailed} in Section~\ref{sec:analysis}.

By Lemma~\ref{lemma:coordinatewise:median}, with probability at least $1-\beta$, all $|{\mu}_i -{\tilde{\mu}}_i| \leq {\sigma}_i$. In this case,
$\|\vectorize{\mu} - \vectorize{\tilde{\mu}}\|_p = O(\|\vectorize{\sigma}\|_p)$.

Let $J = \{ j\in \{1,\dots,n\} \; | \; \|\vectorize{y}^{(j)}\|_2 > C \}$ denote the set of indices of vectors affected by clipping in \algorithmname.
By the triangle inequality and assuming the bound of part 1 holds,
\begin{align*}
	\sum_{j\in J} \left\|(\vectorize{x}^{(j)} - \vectorize{\tilde{\mu}}) \right\|_p
	&\leq \sum_{j\in J} \left\|(\vectorize{x}^{(j)} -\vectorize{\mu}) \right\|_p + |J|\cdot\left\|\vectorize{\mu} - \vectorize{\tilde{\mu}} \right\|_p\\
	&\leq \sum_{j\in J} \left\|(\vectorize{x}^{(j)} -\vectorize{\mu}) \right\|_p + 2 |J|\cdot \|\vectorize{\sigma}\|_p \enspace .
\end{align*}

By assumption~(\ref{assumption:concentrated_generalized}) the probability that $\| \vectorize{x}^{(j)} - \vectorize{\mu} \|_p^p > t \|\vectorize{\sigma}\|_p^p$ is exponentially decreasing in $t$, so setting $t \geq \max\{\ln d, \log(n/\beta))\}$ we have $\| \vectorize{x}^{(j)} - \vectorize{\mu} \|_p = \tilde{O}(\|\vectorize{\sigma}\|_p)$ for all $j=1,\dots,n$ with probability at least $1-\beta$.
Using the triangle inequality again we can now bound
\begin{align*}
\sum_{j\in J} \| \vectorize{x}^{(j)} - \vectorize{\mu} \|_p = \tilde{O}(|J|\cdot\|\vectorize{\sigma}\|_p) \enspace .
\end{align*}

The same line of argument as in Section~\ref{sec:clipping:error} shows that $|J| = \tilde{O}(k + \sqrt{1/\rho})$.
Thus, the clipping error can be bounded by $O((k + \sqrt{1/\rho}) \| \vectorize{\sigma}\|_p)$, and setting $k = n^{-1/2}$ balances the clipping error with the sampling error $\frac{\|\vectorize{\sigma}\|_p}{\sqrt{n}}$.

Lastly, we consider the error due to noise.
First, we find a bound on the clipping threshold $C$.
By the triangle inequality, we may bound
\begin{align*}
\|\vectorize{y}^{(j)}\|_2
& = \left\| (\vectorize{x}^{(j)} - \vectorize{\tilde{\mu}})\, \hat{\Sigma}^{-\tfrac{1}{p+2}} \right\|_2 \\
& \leq \left\| (\vectorize{x}^{(j)} - \vectorize{\mu})\, \hat{\Sigma}^{-\tfrac{1}{p+2}} \right\|_2 + \left\| (\vectorize{\mu} - \vectorize{\tilde{\mu}})\, \hat{\Sigma}^{-\tfrac{1}{p+2}} \right\|_2
\end{align*}

The $i$th coordinate of the first vector $(\vectorize{x}^{(j)} - \vectorize{\mu})\, \hat{\Sigma}^{-1/(p+2)}$ equals $({x}^{(j)}_i - {\mu}_i)/{{\hat{\sigma}}_i}^{2/(p + 2)}$, so assumption (\ref{assumption:scaled_concentrated_generalized}) implies that the length of the first vector is $\tilde{O}(\sqrt{\|\vectorize{\sigma}^{2p/(p + 2)}\|_1})$ with high probability.
By using that the $i$th coordinate of $|\vectorize{\mu} - \vectorize{\tilde{\mu}}| \leq \vectorize{\sigma}$,  $(\vectorize{\mu} - \vectorize{\tilde{\mu}})\,\hat{\Sigma}^{-1/(p+2)}$ has absolute value  $\tilde{O}(\sqrt{\|\vectorize{\sigma}^{2p/(p + 2)}\|_1})$ as well.
The clipping value of $C$ is bounded by the maximum length of a vector $\|\vectorize{y}^{(j)}\|_2$, and thus $C^2 = \tilde{O}(\|\vectorize{\sigma}^{2p/(p + 2)}\|_1)$, with probability at least $1-\beta$.

The scaled noise vector $\vectorize{\eta}\, \hat{\Sigma}^{1/(p+2)}$ has distribution $\mathcal{N}(\vectorize{0},\frac{2C^2}{\rho_3}\hat{\Sigma}^{2/(p+2)})$.
Using ${\hat{\sigma}}_i <  \left({\sigma}_i^{\frac{2p}{p +2 }} + \|\vectorize{\sigma}^{\frac{2p}{p + 2}}\|_1/d\right)^{\frac{p + 2}{2p}}$ and Lemma~\ref{lem:expected:error:gaussian}, we conclude
\begin{align*}
	\mathbf{E}[\|\vectorize{\eta}\, \hat{\Sigma}^{1/(p+2)}\|_p^p]
	& = O\left(\tfrac{2^p C^p}{\rho_3^{p/2}} \sum_{i=1}^d \hat{\Sigma}_{ii}^{p/(p+2)}\right)\\
	& = O\left(\tfrac{2^p C^p}{\rho_3^{p/2}} \sum_{i=1}^d \hat{{\sigma}}_i^{2p/(p+2)}\right)\\
	& = \tilde{O}\left(\frac{2^p\|\vectorize{\sigma}\|_{2p/(p + 2)}^p}{\rho_3^{p/2}}\right) \enspace .
\end{align*}

The result of Theorem~\ref{thm:main-detailed-generalized} is achieved by putting together the different error terms as in the proof of Theorem~\ref{thm:main-detailed}. 
Finally, with probability at most $\beta$ a failure event will happen. In that case, the largest $\ell_p$ error is bounded by the diameter of $[-M, M]^d$, which is $2Md^{1/p}$.
Thus, the contribution to the expected $\ell_p$ error is $2Md^{1/p}\beta$, and can be made $O(1)$ by setting $\beta = \frac{1}{Md^{1/p}}$.
\end{proof}

\section{Choosing The Scaling Value} \label{app:calibration}

Here we discuss the choice of scaling of \algorithmname. %
For parameters $p$ and $\vectorize{\hat{\sigma}} \in \mathbf{R}^d$ we scale each coordinate $i$ by $\vectorize{\sigma}^{-2/(p + 2)}$ on
\cref{alg:recenter-and-scale} of \Cref*{alg:our-algorithm}.
When returning the mean estimate on \cref{alg:plan-mean-estimation} each coordinate $i$ is scaled back by $\vectorize{\sigma}^{2/(p + 2)}$.
The choice of scaling affects how noise is distributed across coordinates.
In this section we consider a simpler setting with independent queries with known sensitivities where we restrict ourselves to adding noise such that there is no clipping error.
We choose the scaling parameters that minimizes the $p$th moment of the noise.

Consider the problem of privately answering $d$ independent real-valued queries with different sensitivities.
We denote the sensitivity of the $i$th query as ${\Delta}_i$ and the vector of all sensitivities as $\vectorize{\Delta}$.
We want to minimize the $p$th moment of the noise.
That is, let $\vectorize{\eta} \in \mathbf{R}^d$ denote the vector where $\vectorize{\eta}_i$ is the noise added to the $i$th query.
We want to minimize $\mathbf{E}[\|\vectorize{\eta}\|_p^p]$.
As discussed in \cref{sec:introduction}, standard approaches either adds a lot of noise for queries with low or high sensitivity as we show here.
We can think of the queries as a $d$-dimensional query which we can release by adding noise from $\mathcal{N}(\vectorize{0},\frac{\bar{\Delta}^2}{2\rho}\mathbf{I})$ where the sensitivity is $\bar{\Delta}=\sqrt{\sum_{i \in [d]} {\Delta}_i^2}=\|\vectorize{\Delta}\|_2$.
This approach adds noise of the same magnitude to each coordinate which might be excessive for queries with very low sensitivity.
The privacy budget is effectively split up between queries weighted by $\vectorize{\Delta}^2$.
Alternatively we could split the privacy budget evenly and answer each query independently by adding noise from $\mathcal{N}(\vectorize{0},d \cdot \frac{\vectorize{\Delta}_i^2}{2\rho})$.
This approach is equivalent to first scaling each coordinate by ${\Delta}_i^{-1}$ and scaling back after adding noise from $\mathcal{N}(\vectorize{0},\frac{d}{2\rho}\mathbf{I})$.
However, this approach adds too much noise to coordinates with high sensitivity.
We want to balance the two approaches and as such \algorithmname uses an in-between scaling value.
We can think of our approach as allocating more privacy budget for queries with high sensitivity, but not so much that it creates large errors for queries with low sensitivity.

Let $\vectorize{x} \in \mathbf{R}^d$ be the answer to the $d$ queries and let $\scalingvec \in \mathbf{R}^d$ be a scaling vector.
Let $\vectorize{y}=\vectorize{x}S^{-1}$ denote the scaled down answers where $S$
is a diagonal matrix with $\scalingvec$ on the diagonal.
We can privately release $\vectorize{y}$ by adding noise from $\mathcal{N}(\vectorize{0}, \frac{\bar{\Delta}^2}{2\rho}\mathbf{I})$ where $\bar{\Delta}^2=\sum_{i \in [d]} {\Delta}_i^2 / s_i^2$.
When scaling back we multiply the noisy answers with $S$.
Our goal is to minimize the $p$th moment.
Specifically, we want to choose $\scalingvec$ to minimize
\[
	\mathbf{E}\left[\| \mathcal{N}(\vectorize{0},\frac{\bar{\Delta}^2}{2\rho}\mathbf{I})S \|_p^p\right] \enspace ,
\]
where $\bar{\Delta}$ and $S$ are defined as above.

\begin{lemma} \label{lem:expected-pth-moment-generic}
	The $p$th moment of the algorithm above is
	\[
		\mathbf{E}\left[\| \mathcal{N}(\vectorize{0},\frac{\bar{\Delta}^2}{2\rho}S^2) \|_p^p\right] = \frac{\Gamma\left(\frac{p+1}{2}\right)}{\rho^{p/2} \sqrt{\pi}} \sum_{i \in [d]} (\bar{\Delta} \cdot \scalingvec_i)^p
	\]
\end{lemma}

\begin{proof}
	It follows from linearity of expectation by summing over each coordinate.
	By Lemma~\ref{lem:expected:error:gaussian} we have
	\begin{align*}
	  \mathbf{E}\left[\vert \mathcal{N}(\vectorize{0},\frac{\bar{\Delta}^2}{2\rho}s_i) \vert^p \right] & = \left(\frac{\bar{\Delta}^2}{2\rho} s_i^2\right)^{p/2} \frac{2^{p/2}\Gamma\left(\frac{p+1}{2}\right)}{\sqrt{\pi}} \\
	  & = \frac{\Gamma\left(\frac{p+1}{2}\right)}{\rho^{p/2} \sqrt{\pi}} (\bar{\Delta} \cdot s_i)^p
	\end{align*}
\end{proof}

Since the fraction outside the sum is not affected by our choice of $\scalingvec$ we just have to find $\scalingvec$ that minimizes $\sum_{i \in [d]} (\bar{\Delta} \cdot s_i)^p$.

\begin{lemma} \label{lem:scale-for-min-error}
	The $p$th moment of the noise is minimized for
	\[
		s_i \propto {\Delta}_i^{2/(p + 2)} .
	\]
\end{lemma}

\begin{proof}
	First notice that multiplying all entries in $\scalingvec$ by the same scalar does not change the result since the magnitude of noise added to $y$ would be scaled by the inverse of the scalar.
	As such, we can restrict our search to $\scalingvec$ such that ${\bar{\Delta}^2=\sum_{i \in [d]} ({\Delta}_i /s_i)^2 = 1}$.
	We also restrict our search to $\scalingvec$ with no negative values since negative values would simply mirror the input and has no impact on the sensitivity.
	As such, any $\scalingvec$ that satisfies those constraints and minimizes $\sum_{i \in [d]}s_i^p$ also minimizes the error. 
	
	We use the method of Lagrange multipliers to find $\scalingvec$. The method is used to find local maxima or minima of a function subject to equality constraints. In our case, it turns our that there is only one minimum which implies that it is the global minimum. We want to minimize $\sum_{i \in [d]} s_i^p$ subject to the restriction $\sum_{i \in [d]} ({\Delta}_i / s_i)^2 - 1 = 0$. The minimum is at the stationary point of the Lagrangian function
	\[
     \mathcal{L}(\scalingvec,\lambda) = \sum_{i \in [d]}s_i^p + \lambda\left(\sum_{i \in [d]} ({\Delta}_i/s_i)^2 - 1\right) ,
    \]
	where $\lambda \in \mathbf{R}$ is the Lagrange multiplier.

	We start by finding the partial derivative of the Lagrangian function with respect to $s_i$.

    \[
     \frac{\partial}{\partial s_i} \mathcal{L}(\scalingvec, \lambda) = \frac{\partial}{\partial s_i} \scalingvec_i^p + \lambda (\Delta_i/\scalingvec_i)^2 = p s_i^{p - 1} - 2\lambda {\Delta}_i^2/s_i^3 \enspace .
    \]

    We then find the root of this function with respect to $s_i$, that is,

	\begin{align*}
		p s_i^{p - 1} - 2\lambda {\Delta}_i^2/s_i^3 & = 0 \\
	s_i^{p + 2} & = 2\lambda {\Delta}_i^2 / p \\
		s_i & = {\Delta}_i^{2/(p+2)} \cdot \gamma \enspace ,
	\end{align*}
	where $\gamma=(2\lambda/p)^{1/(p+2)}$ is a scalar. This finishes the proof.
\end{proof}

\begin{lemma} \label{lem:balanced-pth-moment}
	Let $\scalingvec \in \mathbf{R}^d$ be defined as in Lemma~\ref{lem:scale-for-min-error}.
	Then the $p$th moment of the mechanism described in this section is
	\[
		\mathbf{E}\left[\| \mathcal{N}(\vectorize{0},\frac{\bar{\Delta}^2}{2\rho}\mathbf{I})S \|_p^p\right] = \frac{\Gamma\left(\frac{p+1}{2}\right)}{\rho^{p/2} \sqrt{\pi}} \|\vectorize{\Delta}\|_{2p/(p+2)}^p
	\]
\end{lemma}

\begin{proof}
	From Lemma~\ref{lem:expected-pth-moment-generic} we know that the above equality holds if $\sum_{i \in [d]} (\bar{\Delta} / s_i)^p = \|\vectorize{\Delta}\|_{2p/(p+2)}^p$.
	Continuing the calculations from the proof of Lemma~\ref{lem:scale-for-min-error} we see that for $\bar{\Delta} = 1$ we have

	\begin{align*}
		\sum_{i \in [d]} ({\Delta}_i / s_i)^2 & = \bar{\Delta} \\
		\frac{1}{\gamma^2} \sum_{i \in [d]} {\Delta}_i^{2p/(p+2)} & = 1 \\
		\gamma & = \sqrt{\sum_{i \in [d]} {\Delta}_i^{2p/(p+2)}}
	\end{align*}

	Finally, we find that the $p$th moment is proportional to

	\begin{align*}
		\sum_{i \in [d]} s_i^{p} & = \sum_{i \in [d]} \left({\Delta}_i^{2/(p+2)} \sqrt{\sum_{i \in [d]} {\Delta}_i^{2p/(p+2)}}\right)^p \\
		& = \left(\sum_{i \in [d]} {\Delta}_i^{2p/(p+2)}\right)^{p/2} \sum_{i \in [d]} {\Delta}_i^{2p/(p+2)} \\
		& = \left(\sum_{i \in [d]} {\Delta}_i^{2p/(p+2)}\right)^{1 + p/2} = \|\vectorize{\Delta}\|_{2p/(p+2)}^p \enspace .
	\end{align*}
\end{proof}

\begin{figure}[t]
	\includegraphics[width=0.8\linewidth]{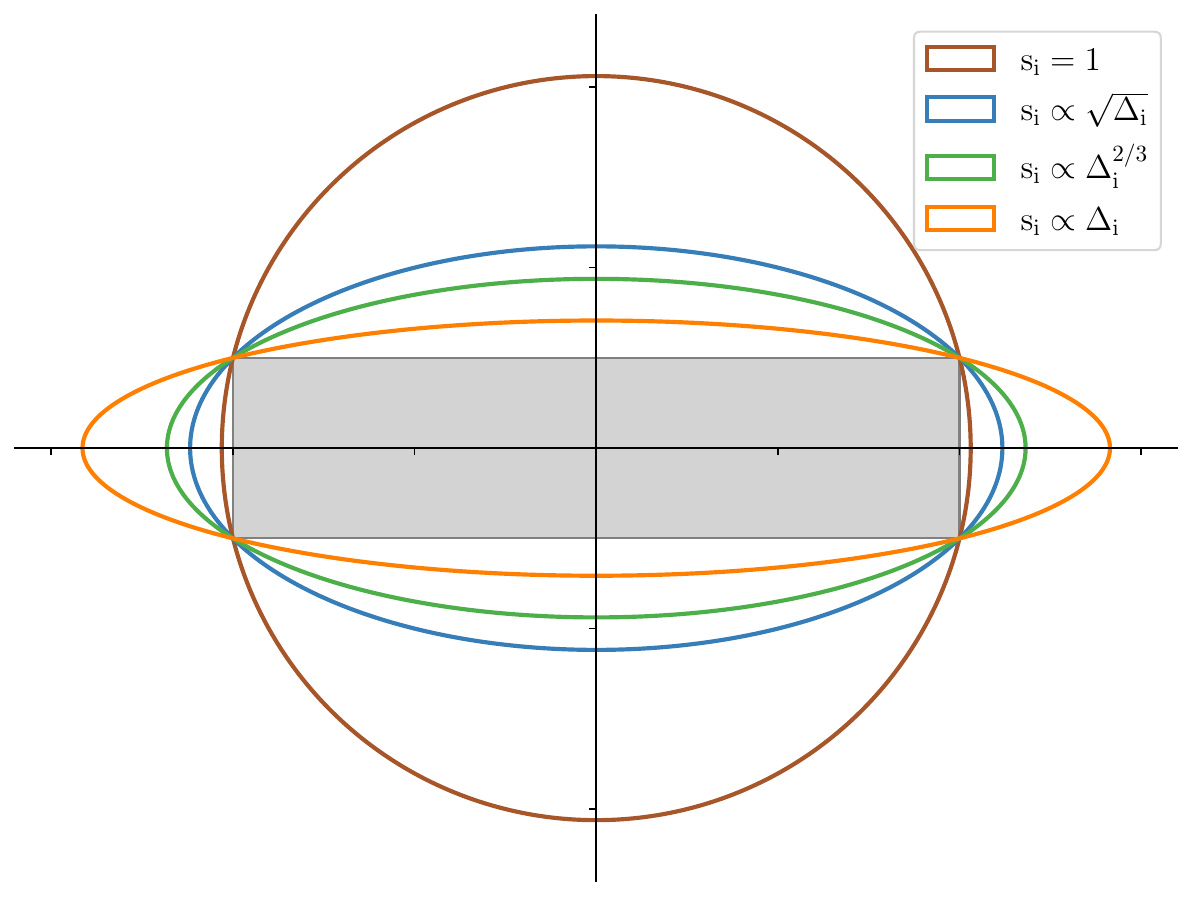}
	\captionsetup{justification=centering}
	\caption{Shape of noise for different scaling parameters.}
	\label{fig:scaling-box}
\end{figure}

Figure~\ref{fig:scaling-box} shows an example of this scaling in two dimensions, where the sensitivity of one query is 4 times as high as the other.
The green and blue lines show the shape of noise when optimizing for the first and second moment, respectively.
In both cases we use most of the privacy budget for the horizontal axis which allows us to add less noise than the orange line where the budget is split evenly.
But we are still able to add less noise to the vertical axis as opposed to the brown line without scaling.

We use the scaling as described above for \algorithmname.
However, in our setting we do not use sensitivities of queries.
Instead we scale each coordinate based on the estimate of the standard deviation.
The intuition is the same: we spend more privacy budget on coordinates with high standard deviation but we still add less noise to coordinates with low standard deviation.

\section{Algorithms for Variance Estimation}\label{app:std-evaluation}
While \algorithmname (Algorithm~\ref{alg:our-algorithm}) assumes that estimates on the standard deviations are known, such estimates have to be computed in a differentially private manner.
Two such ways were described in Section~\ref{sec:experiments} and we will provide more details and empirical results in this section.

We remark that the standard attempt to estimate the variance from a mean estimate $\tilde{\mu}$ is
\begin{align*}
	\left(1/n \sum_{i = 1}^n {x_i^2}\right) - \tilde{\mu}^2
\end{align*}

If $x_i \in [-M, M]$, the sensitivity of this function is $M^2/n$, which, depending on the application, means that too much noise must be added.

\subsection{A Generic Variance Estimation Algorithm}

\begin{algorithm}[thb]
	\caption{\textsc{VarianceEstimate}}
	\label{alg:variance:estimation}
 \begin{algorithmic}[1]
	\State \textbf{Input:} Samples $x^{(1)},\dots,x^{(n)}\in \mathbf{R}$ from $\mathcal{D}$

	\State \textbf{Parameters:} $M, \rho, k$

	\State Split $x^{(1)}, \ldots, x^{(n)}$ into $n' = \lfloor \tfrac{n}{2k}\rfloor$ groups $G_1, \ldots, G_{n'}$.

	\State For each $i \in \{1, \ldots, n'\}$: For $G_i = (\hat{x}_i^{(1)}, \ldots, \hat{x}_i^{(2k)})$, let $\hat{y}^{(i)} = \sum_{j = 1}^k \left(\hat{x}_i^{(2j)} - \hat{x}_i^{(2j + 1)}\right)^2 / 2$.

	\State \textbf{return} $\tilde{\sigma}^2 \leftarrow \frac{\privatequantile_{\rho}^{M^2}(\hat{y}^{(1)},\dots, \hat{y}^{(n')}, 1/2)}{k}.$

	\label{alg:results}
 \end{algorithmic}
 \end{algorithm}

Given a distribution $\mathcal{D}$ with mean $\mu$ and variance $\sigma^2$, we showed in Section~\ref{sec:experiments} that for $X, Y \sim \mathcal{D}$, $\mathbf{E}\left[(X-Y)^2/2\right] = \sigma^2$.
Algorithm~\ref{alg:variance:estimation} is a generalization of the approach used for Gaussian in Section~\ref{sec:experiments}.
For Gaussian data, we made use of the fact that $(X-Y)^2/2$ is $\chi_2$ distributed and it is well-known how to translate an approximate median to an approximate mean.

\begin{lemma}
	Let $\beta > 0, \rho > 0$, and $n = \tilde{\Omega}(\rho^{-1/2})$.
	Let $\mathcal{D}$ be a distribution over $\mathbf{R}$ with mean $\mu$ and variance $\sigma^2 \geq 1$.
	For a constant $\kappa$, assume that $\mathbf{E}[(X-Y)^4] \leq \kappa \sigma^4$.
	With probability at least $1 - \beta$, Algorithm~\ref{alg:variance:estimation} using $k = 16\kappa$ returns an estimate $\hat{\sigma}^2$ such that $\sigma^2/2 \leq \hat{\sigma}^2 \leq 3\sigma^2/2$.
\end{lemma}

\begin{proof}
	Fix a group $G_i$, $1 \leq i \leq n'$.
	Since $\mathbf{E}[(X-Y)^2/2] = \sigma^2$, we know that $\mathbf{E}\left(\hat{y}^{(i)}\right) = k \sigma^2$.
	By Chebychev's inequality,
	\begin{align*}
		\Pr(|\hat{y}^{(i)} - \mathbf{E}[\hat{y}^{(i)}]| > k\sigma^2/2) \leq \frac{\text{Var}(\hat{y}^{(i)})}{(k \sigma^2/2)^2} \leq \frac{k \cdot \mathbf{E}[(X-Y)^4]}{(k \sigma^2/2)^2} \leq \frac{1}{4},
	\end{align*}
	using our assumption on $\mathbf{E}[(X - Y)^4]$ and our choice of $k$.

	As in the proof of Lemma~\ref{lemma:coordinatewise:median},  with probability at least $1 - \exp(-\Omega(n'))$, there are more than $2/3n'$ groups $i$ for which
	\begin{align*}
		\frac{|\hat{y}^{(i)} - \mathbf{E}[\hat{y}^{(i)}]|}{k} \leq \sigma^2/2.
	\end{align*}

	As long as the private quantile selection returns an element $\hat{\sigma}^2$ with rank error at most $n'/6$, $\sigma^2/2 \leq \hat{\sigma}^2 \leq 3\sigma^2/2$. 
	By Lemma~\ref{lemma:approximate-quantile} assuming $n > K \sqrt{\log(M^2)\log(\log(M^2)/\beta)/(2\rho)} = \tilde\Omega(\rho^{-1/2})$, this is true with probability at least $1 - \beta$.
\end{proof}

In contrast to the naïve estimator mentioned above, this estimator has only a logarithmic dependency on the input universe.
In contrast to it, it does not improve from increased sample size above a \emph{minimum sample size} in relation to the parameter $k$ that is necessary to guarantee that the rank error is at most $n/(6 \cdot 2k)$.

By running Algorithm~\ref{alg:variance:estimation} on each coordinate independently with target probability $1 - \beta/d$, and using a union bound, we may summarize:

\begin{corollary}
	\label{cor:variance:estimate}
	Let $\beta > 0, \rho > 0$, and $n = \tilde{\Omega}(\sqrt{d/\rho})$.
	Let $\mathcal{D}$ be a distribution over $\mathbf{R}^d$ with mean $\mu_i$ and variance $\sigma_i^2 \geq 1$ on each coordinate $i \in \{1, \ldots, d\}$.
	For a constant $\kappa$, assume that $\mathbf{E}[(X-Y)^4] \leq \kappa \sigma^4$.
	For each $i$, with probability at least $1 - \beta$, Algorithm~\ref{alg:variance:estimation} on each coordinate using $k = 16\kappa$ returns an estimate $(\hat{\sigma}_1^2,\ldots, \hat{\sigma}_d^2)$  such that $\sigma_i^2/2 \leq \hat{\sigma_i}^2 \leq 3\sigma_i^2/2$.
\end{corollary}

The result from the corollary can be used to find an estimate that \emph{almost} satisfies the conditions in~\eqref{assumption:variance-estimates}. Given $\mathcal{D}$, use the estimator to compute individual estimates $\hat{\sigma_i}^2$. Next, set $\bar{\sigma_i} \gets 2 \hat{\sigma_i} + \|\vectorize{\hat{\sigma}}\|$. With probability at least $1 - \beta$, $\sigma_i \leq \bar{\sigma}_i \leq 2 (\sigma_i + \|\vectorize{\hat{\sigma}}\|_1 )$.
Using such an estimate in the analysis carried out in Section~\ref{sec:noise:error}, the expected error due to noise is increased by a factor $2$.

While hidden in the $\tilde{\Omega}(.)$ notation, only having $n' = n/(2k)$ samples to choose a private quantile might be incompatible with the rank error of the input domain and the dimensionality.
In this case, we can use more groups to ``boost'' $n'$, but we need to adjust $\rho$ because each element is potentially present multiple times.
To cover variances that are smaller than $1$, let $\sigma^2_\text{min}$ be a minimum bound on the variance.
Then, use \privatequantile with $T = \Theta(\log (M/\sigma^2_\text{min}))$ to quantize the input space in steps of $\sigma^2_\text{min}$, which gives a logarithmic depends on $1/\sigma^2_\text{min}$.

\subsection{Variance estimation for Gaussian Data}

In the case that we know that the data is distributed as $\mathcal{N}(\mu, \sigma^2)$, we can tune the variance estimation more towards the distribution as follows. Given an estimate $\tilde{\mu}$ on $\mu$, we can estimate the variance as follows:

$$\tilde{\sigma} \leftarrow \privatequantile_{\rho}^M(x^{(1)},\dots, x^{(n)}, .841) - \tilde{\mu}$$

It is a well-known property of the Gaussian distribution that the .841 quantile is approximately the value $\mu + \sigma$.
However, when estimating the quantile privately we have to adjust for the rank error of the quantile selection, so aiming for this exact quantile may be unwise.

We run the following experiment: we sample $n = 10000$ from $\mathcal{N}(10, \sigma^2)$ with $\sigma^2 \in \{0.001, 1\}$.
For $\rho \in \{10^{-3}, 10^{-2}\}$ we compare (i) three different methods that use different quantiles of
the input data ($.75, .841,$ and  $.9$) to (ii) two different instantiations of \Cref{alg:variance:estimation} for $k = 1$ and $k=4$.
Since we know that $(X-Y)^2$ is $\chi_2$ distributed, we use the mean to median transformation
and divide the approximate median by $(1 - (2/(9k)))^3$.
Each parameter setting is run 100 times and we report on the average relative error $\tfrac{|\hat{\sigma}^2 - \sigma^2|}{\sigma^2}$.
\Cref{table:variance:estimate} reports on empirical results for the variance estimation.
We summarize that \Cref{alg:variance:estimation} is more accurate than direct estimation for both values of $k$, and guarantees very small relative error even for small $\sigma^2$.

\begin{table}[t]
	\begin{tabularx}{0.8\columnwidth}{rrlr}
	\toprule
   $\rho$ & $\sigma^2$ & \textbf{Method} & \textbf{Relative error} \\
   \midrule
    \multirow{10}{*}{0.001} & 0.001 & direct-075 & 0.332 \\
    & 0.001 & direct-0841 &  0.033	\\
    & 0.001 & direct-09 & 0.279 \\
    & 0.001 & general (k=1) & 0.027 \\
    & 0.001 & general (k=4) & \textbf{0.017} \\
    \cmidrule{2-4}
    & 1 & direct-075 & 0.308 \\
    & 1 & direct-0841 & 0.042 \\
    & 1 & direct-09 & 0.314 \\
    & 1 & general (k=1) & 0.025 \\
    & 1 & general (k=4) & \textbf{0.012} \\
   \midrule
   \multirow{10}{*}{0.01} & 0.001 & direct-075 & 0.313 \\
    & 0.001 & direct-0841 & 0.024 \\
    & 0.001 & direct-09 & 0.305 \\
    & 0.001 & general (k=1) & 0.011 \\
    & 0.001 & general (k=4) & \textbf{0.007} \\
    \cmidrule{2-4}
    & 1 & direct-075 & 0.322 \\
    & 1 & direct-0841 & 0.011 \\
    & 1 & direct-09 & 0.265 \\
    & 1 & general (k=1) & 0.020 \\
    & 1 & general (k=4) & \textbf{0.006} \\
   \bottomrule
  \end{tabularx}
   \caption{Comparison of variance estimation algorithms.
   Variants named direct directly use a quantile of the input
   to estimate the standard deviation. Variants named general use
   \Cref{alg:variance:estimation} with $k=1$ and $k=4$, and
   use the median to mean conversion for $\chi_2$ distributed random variables.}
   \label{table:variance:estimate}
\end{table}

\section{Experimental evaluation on correlated data}
\label{app:correlations}

\begin{figure*}[tb]
	\centering
	\includegraphics[width=.8\linewidth]{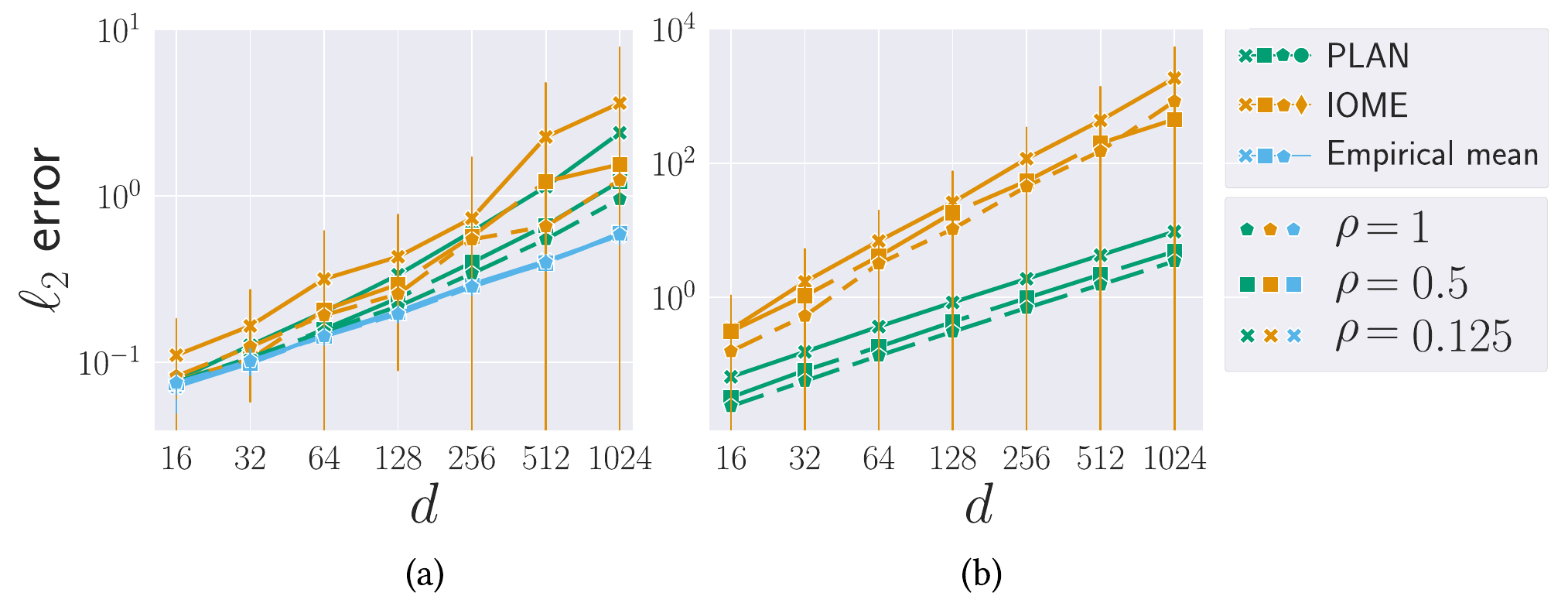}
	\caption{$\ell_2$ error for correlated, synthetic data: (a) \casethem' (random correlation, no skew) and (b) \caseshine' (correlation coefficient .5, skewed data). %
	}
	\label{fig:gaussian-correlations-experiments}
\end{figure*}

The experiments for Gaussian data in~Section~\ref{sec:experiments} focused on diagonal matrices, i.e., uncorrelated vectors.
To validate the trends of our results for correlated data, we design the following additional synthetic test cases:

\begin{itemize}
	\item \casethem': Given $ d \geq 1$, let
     $\Lambda = (\Lambda_{i, j})_{1 \leq i, j \leq d}$ with $$\Lambda_{i, j} \sim \textbf{Uniform}(0, 1/\sqrt{d})$$ 
	for $1 \leq i < j \leq d$.
	Let $\Sigma = \Lambda \cdot \Lambda^T + I_d$, where $I_d$ is the $d \times d$ identity matrix.
	$\Sigma_{i, j}$ is uniformly distributed in $[0, 1)$ for $i \neq j$, and uniform in $[1, 2)$ on the diagonal. 
	Vectors are drawn from $\mathcal{N}(0^d, \Sigma)$, and the setup is identical to \casethem in Table~\ref{tab:experiment-settings}.
	\item \caseshine':  
	Let $\vectorize{\sigma}=(\sigma_1, \ldots, \sigma_d)$ be defined as for the case \caseshine in Table~\ref{tab:experiment-settings}.
	We set $\Sigma = (\Sigma_{i, j})_{1 \leq i, j \leq d}$ with $\Sigma_{i, j} = 0.5 \sigma_i \sigma_j$ for $i \neq j$, and $\Sigma_{i, i} = \sigma_i^2$,
	so we have a correlation coefficient of .5 between the $i$-th and the $j$-th coordinate.
	The remaining parameters are chosen as for \caseshine in Table~\ref{tab:experiment-settings}.	
\end{itemize}

Figure~\ref{fig:gaussian-correlations-experiments} shows the results of the experiments.%
\footnote{We did not evaluate an analogous setting to \caseplausible in Section~\ref{sec:experiments}, because (i) the case $\alpha = 0$ and $\alpha = 2$ is covered by our two experiments, and (ii) the way of generating different $\Sigma$ in the two cases does not allow for a smooth transition from one setting to the other. For example, \casethem' cannot use $\Sigma$ from \caseshine' since $\Sigma$ would not have full rank. }
For non-skewed data with ``random'' correlations (\casethem'), we observe the same trends as for uncorrelated data in Figure~\ref{fig:gaussian-experiments}(a); 
\algorithmname and \instanceoptimalshort provide similar average $\ell_2$ error for a given $d$, with \instanceoptimal having a slightly larger error than in the uncorrelated setting, moving the competitors slightly further away from each than in the uncorrelated setting.

In the setting of skewed data with a constant correlation coefficient, see Figure~\ref{fig:gaussian-correlations-experiments}(b), \algorithmname behaves very similarly to the uncorrelated case, depicted in Figure~\ref{fig:gaussian-experiments}(c). 
In contrast, \instanceoptimalshort has a much larger average $\ell_2$ error compared to its results on uncorrelated data. 
Closer inspection of the results showed that this is likely due to the re-centering step for the individual coordinates. 
Ignoring these outliers for $d = 1024$, the median error of \instanceoptimalshort (subtracting the empirical mean) is $22.61, 45.19, 66.5$ for $\rho = 1, .5, .125$, respectively. 
\algorithmname has median error $3.41, 4.76, 9.40$ in comparison.

\section{Experimental Evaluation of \algorithmname Without Scaling}\label{app:no-scaling}
Section~\ref{sec:generic-bounds} focused on the setting in which we run \algorithmname without variance estimates, i.e., avoid the scaling step.
In this section, we carry out the experiments from Section~\ref{sec:experiments} with this instantiation of \algorithmname as competitor.
From the result of~\Cref{thm:diameter:bound}, we expect this variant of \algorithmname to  behave asymptotically equivalent to \instanceoptimalshort.

\begin{figure*}[tb]
	\centering
	\includegraphics[width=\linewidth]{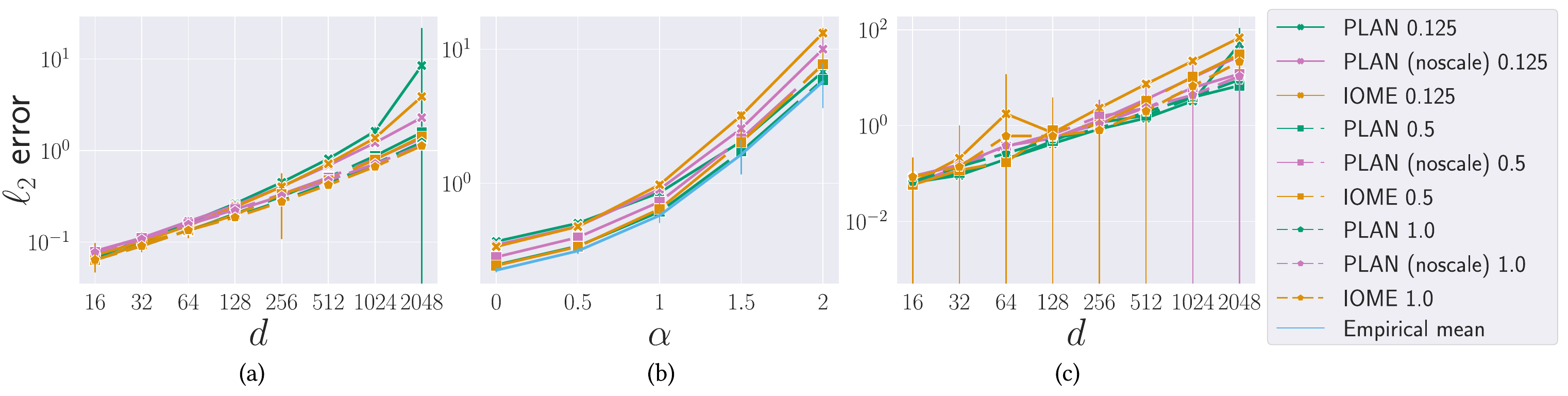}
	\caption{$\ell_2$ error for synthetic Gaussian data using \algorithmname without scaling: We vary (a) dimensions with data without a skew,
	(b) skewness of the variances, and
	(c) dimensions  for skewed data, cf.~Figure~\ref{fig:gaussian-experiments}.}\label{fig:gaussian-no-scaling}
\end{figure*}

\begin{figure*}[tb]
	\centering
	\includegraphics[width=\linewidth]{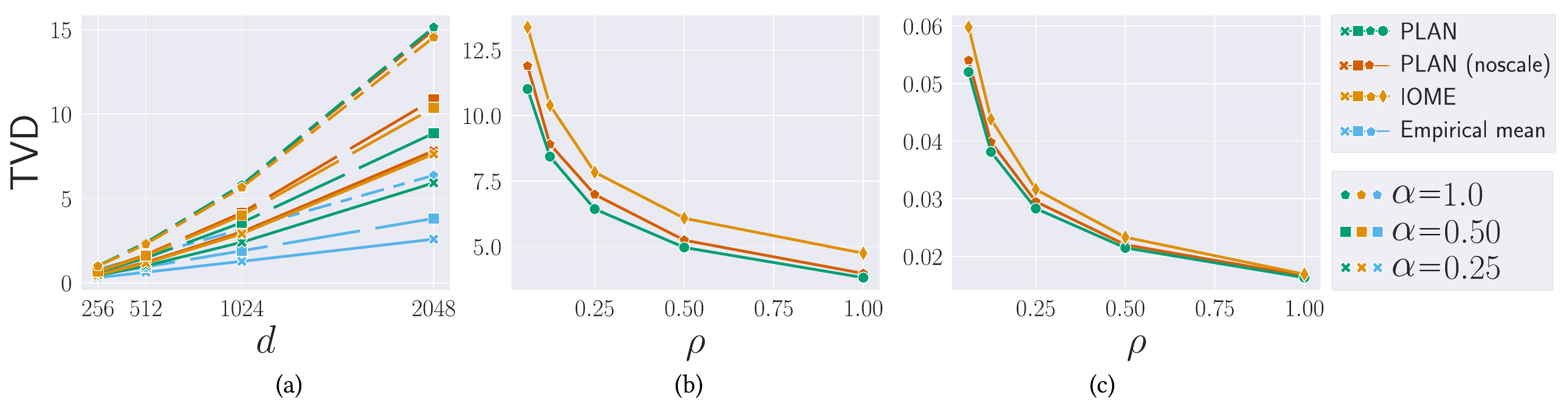}
	\caption{Experiments on synthetic, and real-world binary data.}
	\label{fig:binary-no-scaling}
\end{figure*}

Figure~\ref{fig:gaussian-no-scaling} reports on the results on Gaussian data, comparable with Figure~\ref{fig:gaussian-experiments}.
As we can see from the experiments, \algorithmname without scaling recovers the behavior from \instanceoptimalshort, presenting a slight improvement over \instanceoptimalshort's error.
Unsurprisingly, it improves over \algorithmname in the non-skewed case (which unnecessarily spends privacy budget on an inaccurate variance estimation). However, as soon as data has skew, the coordinate-wise scaling of \algorithmname improves the error.

Figure~\ref{fig:binary-no-scaling} summarizes the empirical utility of \algorithmname without scaling for synthetic, and the real-world binary datasets. As in the Gaussian case, \algorithmname without scaling recovers the behavior from \instanceoptimalshort. It again improves the average error slightly. 
For binary data, their is only a slight gap to the proposed variance-aware scaling.

\section{Empirical Evaluation of \algorithmname with Extended Dimensions}\label{app:gaussian-2048}

In \Cref{fig:gaussian-experiments-2048} we present the plots from the Gaussian case in \Cref{sec:experiment-results}, but here including $d=2048$.
For symmetry we also include \caseplausible in \Cref{fig:gaussian-experiments-2048}~(b), with no changes from the original plot in \Cref{fig:gaussian-experiments}.
In \Cref{fig:gaussian-experiments-2048}~(a) we see \algorithmname and \instanceoptimalshort have comparable accuracy until $d=2048$, when \algorithmname performs worse (for $\rho = 0.125$). This is because for such small $\rho$, the assumptions on the rank error no longer hold.
This is expected behavior as \algorithmname spends some additional budget estimating $\vectorize{\sigma}^2$ in comparison to \instanceoptimalshort.
Similarly in \Cref{fig:gaussian-experiments-2048}~(c) we observe that the error reported by \algorithmname has a high variance due to the rank error increasing. Again, this only happens for the smallest $\rho$ value.

\begin{figure*}[tb]
	\centering
	\includegraphics[width=\linewidth]{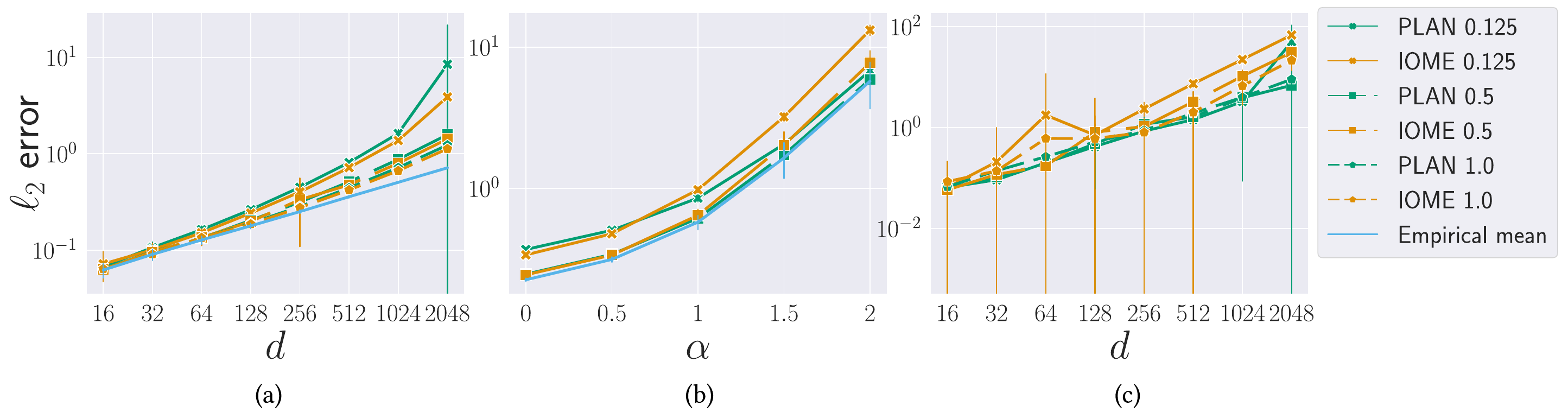}
	\caption{$\ell_2$ error for synthetic Gaussian data when varying (a) dimensions with data without a skew,
		(b) skewness of the variances, and
		(c) dimensions  for skewed data --- note that we compute error relative to the empirical mean rather than the statistical mean in this experiment as sampling error dominates in this setting. Also notice the different scales on the y-axis.}
	\label{fig:gaussian-experiments-2048}
\end{figure*}